\documentclass[final,onefignum,onetabnum]{siamart190516}

\usepackage{lipsum}
\usepackage{amsfonts}
\usepackage{graphicx}
\usepackage{epstopdf}
\usepackage{algorithmic}
\ifpdf
  \DeclareGraphicsExtensions{.eps,.pdf,.png,.jpg}
\else
  \DeclareGraphicsExtensions{.eps}
\fi

\usepackage{amsmath, amsfonts, amssymb, mathtools}
\usepackage{stmaryrd}
\usepackage{bbm}
\usepackage{cases}
\usepackage{hyperref}
\usepackage[T1]{fontenc}
\usepackage[latin1]{inputenc}
\usepackage{upgreek}
\usepackage[english]{babel}
\usepackage[cyr]{aeguill}
\usepackage{xspace}
\usepackage{graphicx}
\usepackage{caption}
\usepackage{subcaption}
\usepackage{color}
\usepackage{epigraph}
\usepackage{etoolbox}
\usepackage{enumitem}
\usepackage{stmaryrd}
\usepackage[margin=1in]{geometry}
\usepackage{cite}
\usepackage{dsfont}
\usepackage{ushort}
\usepackage{xcolor}
\usepackage{tikz}
\usetikzlibrary{arrows,backgrounds}
\usepgflibrary{shapes.multipart}
\usetikzlibrary{arrows.meta}

\def\dbE{\mathbb{E}}

\def\dbY{\mathbb{G}}

\def\dbN{\mathbb{N}}

\def\dbR{\mathbb{R}}

\def\dbQ{\mathbb{Q}}
\def\dbY{\mathbb{Y}}

\def\cC{{\cal C}}

\def\cE{{\cal E}}
\def\cF{{\cal F}}

\def\cI{{\cal I}}

\def\cK{{\cal K}}

\def\cP{{\cal P}}

\def\cR{{\cal R}}
\def\cS{{\cal S}}

\def\cW{{\cal W}}

\def\g{\gamma}
\def\d{\delta}
\def\e{\varepsilon}

\def\k{\kappa}
\def\l{\lambda}

\def\si{\sigma}

\def\f{\varphi}
\def\th{\theta}
\def\o{\omega}

\def\D{\Delta}
\def\Si{\Sigma}
\def\Th{\Theta}

\def\eps{\epsilon}
\def\sfm{{\mathsf m}}
\def\cd{\cdot}
\def\cds{\cdots}
\def\bp{\bar\pi}
\def\sfy{{\mathsf y}}
\def\rc{{\rm rc}}
\def\syc{{\rm sc}}

\def\ol{\overline}

\def\q{\quad}
\def\no{\noindent}

\def\ms{\vspace{3mm}}

\DeclareMathOperator*{\argmin}{\arg\!\min}
\DeclareMathOperator*{\argmax}{\arg\!\max}

\newcommand{\be}{\begin{equation}}
\newcommand{\ee}{\end{equation}}
\newcommand{\bee}{\begin{equation*}}
\newcommand{\eee}{\end{equation*}}
\newcommand{\bea}{\begin{eqnarray}}
\newcommand{\eea}{\end{eqnarray}}
\newcommand{\beaa}{\begin{eqnarray*}}
\newcommand{\eeaa}{\end{eqnarray*}}

\newsiamremark{remark}{Remark}
\newsiamremark{hypothesis}{Hypothesis}
\crefname{hypothesis}{Hypothesis}{Hypotheses}
\newsiamthm{assum}{Assumption}
\newsiamthm{eg}{Example}

\headers{Game on Random Environment,  MFL System and Neural Networks}{G. Conforti, A. Kazeykina, and Z. Ren}

\title{Game on Random Environment,  Mean-field Langevin System and Neural Networks}

\author{Giovanni CONFORTI \thanks{Centre de Math\'ematiques Appliqu\'ees, Ecole Polytechnique, IP Paris, 91128 Palaiseau Cedex, France (\email{giovanni.conforti@polytechnique.edu}).}
\and Anna KAZEYKINA \thanks{Laboratoire de Math\'ematiques d'Orsay, Universit\'e
Paris-Saclay, 91405 Orsay, France 
  (\email{anna.kazeykina@math.u-psud.fr}).}
\and Zhenjie REN \thanks{Ceremade, Universit\'e Paris-Dauphine, PSL Research University, 75016 Paris, France (\email{ren@ceremade.dauphine.fr}).}}

\usepackage{amsopn}

\begin{document}

\maketitle

\begin{abstract}
In this paper we study a type of games regularized by the relative entropy, where the players' strategies are coupled through a random environment variable. Besides the existence and the uniqueness of equilibria of such games, we prove that the marginal laws of the corresponding mean-field Langevin systems can converge towards the games' equilibria in different settings. As applications, the dynamic games can be treated as games on a random environment when one treats the time horizon as the environment. In practice, our results can be applied to analysing the stochastic gradient descent algorithm for deep neural networks in the context of supervised learning as well as for the generative adversarial networks. 
\end{abstract}



\section{Introduction}

The approximation of the equilibria is at the heart of the game theory. The classic literature introduces a natural relation connecting the equilibria of the games and the optima of the sequential decision problems.  This leads to a fruitful research on the topics such as approachability, regret and calibration, see the survey by V. Perchet \cite{perchet14} and the books by N. Cesa-Bianchi and G. Lugosi \cite{CL06} and by D. Fudenberg and D. K. Levine \cite{FL98}. In particular, the gradient-based strategy often plays a crucial rule in approximating the equilibria. In the present paper we study the analog to the gradient-based strategy in the continuous-time setting, namely, we aim at approximating the equilibria of games using the diffusion processes encoded with the gradients of potential functions. 

Consider a game with $n$ players. The mixed Nash equilibrium is defined to be a collection of probability measures $(\nu^{*,i})_{i=1,\cds, n}$ such that
\be\label{nashclassic}
\nu^{*, i} \in \argmin_{\nu^i} \int f^i(x^1 ,\cds, x^n) \nu^i(dx^i) \prod_{j\neq i}\nu^{*,j}(dx^j),
\ee
which means that each player can no longer improve his performance  by making a unilateral change of strategy. Note that in this classical setting, the potential function of each player 
\bee
\nu^i\mapsto  F^i\big(\nu^i, (\nu^j)_{j\neq i}\big):=\int f^i(x^1 ,\cds, x^n) \nu^i(dx^i) \prod_{j\neq i}\nu^{j}(dx^j)
\eee
is linear. In this paper we shall allow the potential function to be {\it nonlinear} in view of the applications, in particular, to the neural networks (see \cref{sec:app}).
As another generalization to  the classic theory, we consider games on a {\it random environment}. Introduce a space of environment $\dbY$ and fix a probability measure $\sfm$ on it. We urge each player to choose a strategy among the probability measures $\nu^i$ on the product space $\dbR^{n^i}\times \dbY$ such that the marginal law of $\nu^i$ on $\dbY$, $\nu^i_Y$, matches the fixed distribution $\sfm$. Typically, in our framework we  consider the game of which the Nash equilibrium is  a collection of probability measures $(\nu^{*,i})_{i=1,\cds, n}$ on the product spaces such that
\be\label{nashcor}
\nu^{*, i} \in \argmin_{\nu^i : \nu^i_Y=\sfm} \int f^i(x^1 ,\cds, x^n, \sfy) \nu^i(dx^i|\sfy) \prod_{j\neq i}\nu^{*,j}(dx^j|\sfy) \sfm(d\sfy) +\frac{\si^2 }{2} H(\nu^i | {\rm Leb}\times \sfm ),
\ee 
where $\nu(\cd|\sfy)$ denotes the conditional probability given $\sfy$, and we add the relative entropy $H$ as a regularizer.
In contrast to the conventional definition of the Nash equilibrium \cref{nashclassic}, where the players' strategies are uncorrelated, in our setting the strategies of the players are allowed to be coupled through the environment. Moreover,  the general framework of the present paper goes beyond the particular game \cref{nashcor}, by allowing the cost function to be nonlinear in $(\nu^i)_{i=1,\cds, n}$. As an application,  we observe (\cref{eg:dynamicgame}) that relaxed dynamic games can be viewed as games on random environment, where the environment $\dbY$ is the time horizon.

One of our main contributions is the first order condition of the optimization on the probability space given a marginal constraint (\cref{thm:FOC}), which naturally provides a necessary condition for being a Nash equilibrium of a game on  random environment (\cref{cor:NC_Nash}). This result is a generalization to the first order condition in Proposition 2.4 in \cite{HRSS19} for the optimization on the probability space without marginal constraint. The key ingredient for this analysis is the linear functional derivative $\frac{\d F}{\d \nu}$, first introduced for the variational calculus and recently popularized by the study on the mean-field games, see e.g. Cardaliaguet et al. \cite{cardaliaguet2015master}, Delarue et al. \cite{DLR19, DLR19II}, Chassagneux et al. \cite{CLTse19}. Roughly speaking, we prove that 
\begin{multline}\label{intro:FOC}
\mbox{if $\nu^*\in \argmin_{\nu: \nu_Y =\sfm} F(\nu) + \frac{\si^2}{2}H(\nu| {\rm Leb}\times \sfm)$,}\\
\mbox{then $\nabla_x  \frac{\d F}{\d \nu}(\nu^*, x, \sfy) + \frac{\si^2 }{2} \nabla_x \ln \nu^*(x|\sfy) = 0$, for all $x$, $\sfm$-a.s. $\sfy$. }
\end{multline}
Besides the first order condition for the Nash equilibrium, we also provide sufficient conditions on the linear functional derivative so that the game on a random environment admits a (unique) equilibrium.

The first order equation in \cref{intro:FOC}  clearly links the minimizer $\nu^*$ (or the Nash equilibrium in the context of games) to the invariant measure  of a system of  diffusion processes, see \cref{eq:MFL} below. Since the dynamics of the diffusion processes depends on their marginal distributions (in other word, McKean-Vlasov diffusion, see \cite{Meleard96, Sznit89}) and involves the gradients of the potential functions, we name the system mean-field Langevin  (MFL) system. Further, we study the different settings where the marginal laws of the MFL system converge to the unique invariant measure, which, due to the first order condition, must coincide with the Nash equilibrium of the game on random environment. 
In the case with small dependence on the marginal laws, we prove the exponential ergodicity of the MFL system, using the reflection coupling (see \cite{eberle11, eberle2019quantitative}). 
In the  case of one-player game (in other word, optimization), once the potential function is convex, we use an argument, similar to that in \cite{HRSS19, HKR19}, based on the Lasalle's invariant principle to prove the (non-exponential) ergodicity of the MFL system under  mild conditions on the coefficients. 

In view of applications, our result can be used to justify the applicability of the gradient descent algorithm for training (deep) neural networks. As mentioned in \cite{JSS19, mei2018mean, mei2019mean, HRSS19, HKR19}, the supervised learning with (deep) neural networks can be viewed as a minimization problem (or optimal control problem in the context of deep learning) on the space of probability measures, and the gradient descent algorithm is approximately a discretization of the corresponding mean-field Langevin dynamics. The present paper provides a more general framework for such studies. In particular, it is remarkable that in \cref{sec:app} we provide a theoretically convergent numerical scheme for the generative adversarial networks (GAN) as well as a way to characterize the training error. 

The rest of the paper is organized as follows. In \cref{sec:notation} we introduce the definitions of a game on a random environment and the corresponding MFL system.  In \cref{sec:main}  the main theorems of the paper are stated without proofs. In \cref{sec:app} we present the applications to the dynamic games and the GAN. Finally in \cref{sec:proof} we present the proofs of the main theorems.

\section{Notation and definitions}\label{sec:notation}

%
%

\subsection{Preliminary}

Denote by $\dbY$ the space of environment, and assume $\dbY$ to be Polish. Throughout the paper, we fix a probability measure $\sfm$ on $\dbY$. Define the product space $\bar \dbR^d:= \dbR^d\times \dbY$ for $d\in \dbN$.  In this paper we consider a game in which the players choose strategies among $\Pi:=\{\bp\in \cP_p(\bar\dbR^d): \bp(\dbR^d,\cd) = \sfm \}$, where $\cP_p(\bar\dbR^d)$ is the space of probability measures on $\bar\dbR^d$ with finite $p$-moments for some $p\ge 1$.  We say that a function $F:\Pi\rightarrow \dbR$ is in $\cC^1$ if there exists a  function $\frac{\d F}{\d \nu}: \Pi\times \bar\dbR^d \rightarrow \dbR$ such that for all $\nu,\nu'\in \Pi$
\be\label{C1derv}
F(\nu') - F(\nu) = \int_0^1 \int_{\bar\dbR^d} \frac{\d F}{\d \nu}\big((1-\l)\nu+\l \nu',  \bar x\big)(\nu'-\nu)(d \bar x)d\l.
\ee
We will refer to $\frac{\d F}{\d \nu}$ as the linear functional derivative.
There is at most one $\frac{\d F}{\d \nu}$, modulo a constant, satisfying \cref{C1derv}.

 Here is the basic assumption we apply throughout the paper.

\begin{assum}[basic assumption]\label{assum:base}
Assume that for some $p\ge 1$, the function $F: \Pi\rightarrow \dbR$ belongs to $\cC^1$ and
\begin{itemize}
\item $F$ is $\cW_p$-continuous, where $\cW_p$ stands for the $p$-Wasserstein distance;

\item $\frac{\d F}{\d \bar\pi}: (\bar\pi, x, \sfy)\in\Pi\times\dbR^d\times\dbY\mapsto \frac{\d F}{\d \bar\pi}(\bar\pi, x, \sfy) \in \dbR$ is $\cW_p$-continuous in $\bp$ and  continuously differentiable in $x$;


\item $\frac{\d F}{\d \bar\pi}$ is of $p$-polynomial growth in $\bar x= (x, \sfy)$, that is, $\sup_{\bp\in \Pi}\big| \frac{\d F}{\d \bar\pi}(\bp, \bar x) \big| \le C(1+|\bar x|^p)$.
\end{itemize}
\end{assum}

\begin{remark}
 Since in our setting the law on the environment $\dbY$ is fixed, by disintegration we may identify a distribution 
$\bar \pi\in \Pi$
with the probability measures $\big(\pi(\cdot | \sfy)\big)_{\sfy\in\dbY}\subset \cP_p(\dbR^d)$ such that
$\bar \pi(d\bar x) = \pi(dx | \sfy) \sfm(d\sfy)$.
\end{remark}

\subsection{Game on random environment}

In this paper, we consider a particular game in which the strategies of the $n$ players are correlated through the random environment (or signal) $\dbY$. 
 Let $n^i\in \dbN$ for $i=1,\cdots, n$ and $N:=\sum_{i=1}^n n^i$.  
 As mentioned before, the $i$-th player chooses his strategy (a probability measure) among $\Pi^i:=\{ \nu\in \cP_p(\bar\dbR^{n^i}) : \nu(\dbR^{n^i} ,\cd) = \sfm \}$, while the joint distribution of the other players' strategies belongs to the space $\Pi^{-i}:=\{ \nu\in \cP_p(\bar\dbR^{N-n^i}) : \nu(\dbR^{N-n^i} ,\cd) = \sfm \}$.
 The $i$-th player aims at optimizing his objective function $F^i: \Pi^i\times\Pi^{-i}\rightarrow \dbR$. More precisely, he faces the optimization:
\bee
\mbox{given $\mu\in \Pi^{-i}$},\q\q\q
\mbox{solve}\q\inf_{\nu\in \Pi^i} F^i(\nu, \mu).
\eee
In this paper, we are more interested in solving a regularized version of the game above. We use the relative entropy with respect to ${\rm Leb}^{n^i}\times\sfm$, denoted by $H^i$, as the regularizer. Namely, given $\mu\in \Pi^{-i}$ the $i$-th player solves:
\bee
\q\inf_{\nu\in \Pi^i} V^i(\nu, \mu), \q V^i(\nu, \mu):= F^i(\nu, \mu) +\frac{\si^2}{2} H^i(\nu)\q\mbox{for some $\si>0$.}
\eee
For $\bar\pi\in \Pi$, we denote by $\bar\pi^i\in \Pi^i$ its marginal distribution on $\bar\dbR^{n^i}$, and by $\bar\pi^{-i}\in \Pi^{-i}$ the marginal distribution on  $\bar\dbR^{N-n^i}$.

\begin{definition}\label{defn:nash}
A probability measure $\bar\pi \in \Pi$ is a Nash equilibrium of this game, if 
\bee
\bar\pi^i \in \arg\min_{\nu\in \Pi^i} V^i(\nu, \bar\pi^{-i}), \q\q\mbox{for all}\q i=1,\cdots, n.
\eee
\end{definition}

\begin{eg}\label{eg:dynamicgame}
To have a concrete example of games on random environment, we refer to the dynamic games, both discrete-time and continuous-time models. In the discrete-time case, let $\dbY:=\{1, \cdots, T\}$ for some $T\in \dbN$ and  $\sfm$ be the uniform distribution on $\dbY$. Define the controlled dynamics:
\bee
\Theta_\sfy^i =   \f^i_\sfy (\Theta_{\sfy-1}^i, \pi^i(\cd| \sfy), \bar \pi^{-i}), \q\mbox{where}\q \bar\pi^i(\cdot, \sfy) =  \pi^i(\cd|\sfy) \sfm(\sfy), \q\mbox{for}\q \sfy\in \dbY.
\eee
If the $n$ players minimize the objective functions of the form $f^i\big((\Theta^i_\sfy)_{\sfy\in \dbY}\big)$ by choosing the strategy $\bar\pi^i$, then the game fits the framework of this paper. 

Similarly for the continuous-time model, consider the space $\dbY:=[0,T]$ for $T\in \dbR$ and let $\sfm$ be the uniform distribution on the interval. Define the continuous-time dynamics:
\bee
d\Theta^i_\sfy = \f^i \big( \pi^i(\cd|\sfy), \bar \pi^{-i}, \Theta_\sfy^i, \sfy\big)d\sfy, \q\mbox{where}\q \bar\pi^i (\cdot, d \sfy) = \pi^i(\cd|\sfy) \sfm(d\sfy), \q\mbox{for}\q \sfy\in \dbY.
\eee
If the $n$ players  minimize the objective functions of the form $f^i\big((\Theta^i_\sfy)_{\sfy\in \dbY}\big)$ by choosing the strategy $\bar\pi^i$, then this game also fits in the  framework discussed above.
\end{eg}

\subsection{Mean-field Langevin system}

For any fixed $\mu\in \Pi^{-i}$, we assume that $F^i(\cdot, \mu): \nu\in \Pi^i \mapsto F^i(\nu,\mu)\in  \dbR$ satisfies \cref{assum:base}. The linear derivative is denoted by $ \frac{\d F^i}{\d \nu}(\cd, \mu, \cd):  (\nu, \bar x^i) \mapsto  \frac{\d F^i}{\d \nu} (\nu , \mu, \bar x^i)$, with $\bar x^i= (x^i, \sfy)\in \bar\dbR^{n^i}$.  In order to compute Nash equilibria of the game on random environment,  we are interested in the following mean-field Langevin (MFL) dynamics:
\be\label{eq:MFL}
d X^i_t = - \nabla_{x^i} \frac{\d F^i}{\d \nu} (\bar\pi^i_t , \bar\pi^{-i}_t, X^i_t, Y)dt + \si dW^i_t, \q\mbox{for}\q i=1,\cdots,n,
\ee
where $W=(W^i)_i$ is an $N$-dimension Brownian motion, $Y$ is a random variable taking values in $\dbY$ and satisfying the law $\sfm$, and $\bar\pi_t := {\rm Law}\bar X_t$ with $\bar X_t :=(X^1_t, \cdots, X^n_t, Y)$. 
 In this paper we will discuss the relation between the MFL dynamics and the Nash equilibrium of the game on the random environment. 

\begin{remark}
Here are some important observations:
\begin{itemize}

\item The random variable $Y$ plays the role of parameter in the MFL system. This leads us to study the system:
\be\label{eq:MFL_sys}
d X^\sfy_t = - \left(\nabla_{x^i} \frac{\d F^i}{\d \nu} (\bar\pi^i_t , \bar\pi^{-i}_t, X^{\sfy,i}_t, \sfy)\right)_{i=1,\cds,n}dt
		 + \si dW_t, \q\mbox{for}\q  \sfm\mbox{-a.s.}~\sfy\in \dbY.
\ee
Formally, the marginal laws of the MFL system above with a fixed $y\in \dbY$ satisfy the following system of Fokker-Planck equations:
\begin{multline}
\label{eq:FP}
\partial_t \pi^i(\cd|\sfy) = \nabla_{x^i}\cdot \left(\nabla_{x^i} \frac{\d F^i}{\d \nu}(\bar\pi^i, \bar\pi^{-i}, \cdot, \sfy)\pi^i(\cd|\sfy)   +\frac{\si^2}{2} \nabla_{x^i} \pi^i(\cd|\sfy)\right), \q \\
\mbox{for all $i=1, \cds,n$, $\sfm$-a.s. $\sfy\in \dbY$.}
\end{multline}

\item For fixed $\sfy\in \dbY$, the dynamic systems for $\big(X^i(\sfy)\big)_i$ are only weakly coupled through the marginal distributions.

\item Although we name the system after Langevin, the drift term of the dynamics of the aggregated vector $ (X^i)_{i=1,\cdots, n}$ is in general not in the form of the gradient of a potential function.
\end{itemize}
\end{remark}

\section{Main results}\label{sec:main}

\subsection{Optimization with marginal constraint}

One of our observations  is the following first order condition of the optimization over the probability measures with marginal constraint. 

\begin{theorem}[first order condition]\label{thm:FOC}
 Let $F: \Pi\rightarrow \dbR$ satisfy \cref{assum:base}. Define $V(\bar\pi):=F(\bar\pi) + \eta H(\bar\pi)$ for some $\eta>0$.  If $\bar\pi^*\in \argmin_{\bar\pi\in \Pi} V(\bar\pi)$, then
\be\label{eq:necessary}
\nabla_x \frac{\d F}{\d \bar\pi}(\bar\pi^*, \cd,\sfy) + \eta \nabla_x\ln\big(\pi^*(\cd|\sfy)\big) =0\q\mbox{for $\sfm$-a.s. $\sfy$.}
\ee
Conversely, if we additionally assume that $F$ is convex, then $\bp^*\in \Pi$ satisfying \cref{eq:necessary} implies 
$\bar\pi^*\in \argmin_{\bar\pi\in \Pi} V(\bar\pi)$.
\end{theorem}

\begin{remark}
We remark that 
\begin{itemize}
\item  the regularizer $H(\bp)$ plays an important role for the proof of the necessary condition. Without it, for  $\bar\pi^*\in \argmin_{\bar\pi\in \Pi} V(\bar\pi)$  we can only conclude that there is a measurable function $f:\dbY\rightarrow\dbR$ such that
\bee
\frac{\d F}{\d \bar\pi}(\bar\pi^*, x,\sfy) =f(\sfy) ,\q \bp^*\mbox{-a.s.};
\eee

\item for the readers more interested in the minimization of the unregularized potential function $F$, by standard argument (see e.g. \cite[Proposition 2.3]{HRSS19}) one may prove that under mild conditions the minimum of $F+\eta H$ converges to the minimum of $F$ as $\eta\rightarrow 0$.

\end{itemize}
\end{remark}

\subsection{Equilibria of games on random environment}

 Applying the first order condition above to the context of the game on  random environment, we immediately obtain the following necessary condition for the Nash equilibria.

\begin{corollary}[Necessary condition for Nash equilibria]
\label{cor:NC_Nash}
For $i=1,\cdots, n$ and $\mu\in \Pi^{-i}$, let $F^i(\cdot, \mu): \nu \in\Pi^i\mapsto F^i(\nu, \mu)\in \dbR$  satisfy \cref{assum:base}. If  $\bar\pi\in \Pi$ is a Nash equilibrium, we have for  $i=1,\cdots,n,$
\be\label{eq:FOCeq}
\nabla_{x^i} \frac{\d F^i}{\d \nu}(\bar\pi^i, \bar\pi^{-i}, x^i,\sfy) + \frac{\si^2}{2} \nabla_{x^i}\ln\big(\pi^i(x^i|\sfy)\big) =0\q\mbox{for all $x^i\in \dbR^{n^i}$ and $\sfm$-a.s. $\sfy\in \dbY$}.
\ee
\end{corollary}

We shall use  the first order equation \cref{eq:FOCeq} to show the following sufficient condition for the uniqueness of Nash equilibrium.

\begin{corollary}[Uniqueness of Nash equilibrium: Monotonicity]
\label{cor:uniq_monot}
The functions $(F^i)_{i=1,\cdots, n}$ satisfy the monotonicity condition, if for $\bar\pi, \bar\pi'\in \Pi$ we have
\bee
\label{eq:monotonicity}
\sum_{i=1}^n \int \left( \frac{\d F^i}{\d \nu}(\bar\pi^i, \bar\pi^{-i}, \bar x^i) - \frac{\d F^i}{\d \nu}(\bar\pi'^i, \bar\pi'^{-i}, \bar x^i)  \right)(\bar\pi - \bar\pi')(d\bar x) \ge 0.
\eee
We have the following results:
\begin{itemize}
\item[{\rm (i)}]  for $n=1$,  if a function $F$ satisfies the monotonicity condition then it is convex on $\Pi$. Conversely, if $F$ is convex and  satisfies \cref{assum:base}, then $F$ satisfies the monotonicity condition.

\item[{\rm (ii)}] in general ($n\ge 1$), for $i=1,\cdots, n$  and any $\mu\in\Pi^{-i}$, let $F^i(\cdot, \mu): \nu \in\Pi^i \mapsto F^i(\nu, \mu)\in \dbR$  satisfy \cref{assum:base} and $ (F^i)_{i=1, \cds, n} $ satisfy the monotonicity condition. Then for any two Nash equilibria $ \bar \pi^* , \bp'^* \in \Pi $ we have $(\bp^*)^i =(\bp'^*)^i$ for all $i=1, \cds,n$. 
\end{itemize}
\end{corollary}

\begin{remark}
Similar monotonicity conditions are common assumptions to ensure the uniqueness of equilibrium in the game theory, in particular in the literature of mean-field games, see e.g. \cite{lasry2007}. 
\end{remark}

As for the existence of Nash equilibria, we obtain the following result following the classical argument based on the fixed point theorem. 

\begin{theorem}[Existence of equilibria]\label{thm:exist_nash}
Assume  that  for $i=1,\cdots, n$,  and $\mu\in \Pi^{-i}$ 
\begin{itemize}
\item[{\rm (i)}]  the set $\argmin_{\nu\in \Pi^i} V^i(\nu, \mu) $ is  non-empty and convex;

\item[{\rm (ii)}]  the function $\tilde F^i (\bp):= F^i(\bp^i, \bp^{-i})$ is  $\cW_p$-continuous on $\Pi$;  

\item[{\rm (iii)}]  the function $F^i(\cdot, \mu): \nu \in\Pi^i \mapsto F^i(\nu, \mu)\in \dbR$  satisfies \cref{assum:base},  and there exist some $q\ge q'>0,~ C,C'>0 \in \dbR$ such that for all $\bp\in \Pi$ we have
\bea\label{basic_dissipative}
C'| \bar x^i |^{q'} - C\le \frac{\d F^i}{\d \nu}(\bp^i, \bp^{-i},  \bar x^i) \le C| \bar x^i |^q + C.
\eea
\end{itemize}
Then there exists at least one Nash equilibrium $\bp^*\in \Pi$ for the game on random environment.
\end{theorem}

\begin{remark}
There are various sufficient conditions so that  the set $\argmin_{\nu\in \Pi^i} V^i(\nu, \mu)$ is convex, for example, the function $\nu\mapsto V^i(\nu, \mu)$ is quasi-convex, or  $V^i(\nu, \mu)$ has a unique minimizer. That is why we leave the assumption (i) in the abstract form.
\end{remark}

\subsection{Invariant measure of the MFL system}

In view of the Fokker-Planck equation \cref{eq:FP}, the first order equation \cref{eq:FOCeq} appears to be a sufficient condition for $ \bar \pi$ being an invariant measure of the MFL system \cref{eq:MFL}. That is why we consider the MFL dynamics as a reasonable tool to compute the Nash equilibria of the game on random environment.

The following \cref{thm:well} suggests that proving the {\it existence} of Nash equilibria  and the {\it uniqueness} of the invariant measure, we can establish the equivalence between the invariant measure of  \cref{eq:MFL} and one Nash equilibrium.
While the existence of Nash equilibria has been discussed in\cref{thm:exist_nash}, 
the uniqueness of invariant measure of mean-field dynamics is more complicated and is indeed a long-standing problem in probability and analysis. We are going to use the coupling argument in order to obtain the contraction result in \cref{thm:contraction}. 

Define the average Wasserstein distance:
\bee
\ol\cW_p(\bp, \bp') : =\Big( \int_\dbY \cW_p^p\big(\pi(\cd|\sfy), \pi'(\cd|\sfy)\big) \sfm(d\sfy) \Big)^\frac1p,
\eee
and  the spaces of flow of probability measures:
\beaa
& C_p([0,T], \Pi) : =\left\{(\bp_t)_{t\in [0,T]}: ~ \mbox{for each $t$,}~ \bp_t\in \Pi,~\mbox{and}~ t \mapsto \bp_t ~\mbox{is continuous in $\cW_p$}\right\},\\
& \ol C_p([0,T], \Pi) : =\left\{(\bp_t)_{t\in [0,T]}: ~ \mbox{for each $t$,}~ \bp_t\in \Pi,~\mbox{and}~ t \mapsto \bp_t ~\mbox{is continuous in $\ol\cW_p$}\right\}.
\eeaa

\begin{theorem}\label{thm:well}
For $i=1,\cdots, n$ and $\mu\in \Pi^{-i}$, let $F^i(\cdot, \mu): \nu \in\Pi^i\mapsto F^i(\nu, \mu)\in \dbR$  satisfy \cref{assum:base}. Further assume that
\begin{itemize}
\item the initial distribution $\bp_0={\rm Law}(\bar X_0)\in \Pi$;
\item for each $i=1,\cds,n$, the function $\nabla_{x^i} \frac{\d F^i}{\d \nu}$ is Lipschitz continuous in the following sense
\begin{multline*}
\left| \nabla_{x^i} \frac{\d F^i}{\d \nu}(\nu, \mu,  x^i, \sfy) -\nabla_{x^i} \frac{\d F^i}{\d \nu}(\nu', \mu', x'^i, \sfy)\right|\\
\le C\left(\cW_p(\nu, \nu') + \cW_p (\mu, \mu')  + | x^i -  x'^i |\right)
+ C_0 \cW_p \big(\mu(\cd|\sfy), \mu'(\cd|\sfy)\big),
\end{multline*}
and satisfies
\be\label{eq:uniform0}
\sup_{\nu\in \Pi^i, \mu\in \Pi^{-i}, \sfy\in \dbY} \left| \nabla_{x^i} \frac{\d F^i}{\d \nu}(\nu, \mu,  0, \sfy) \right| <\infty.
\ee
\end{itemize}
Then the MFL system \cref{eq:MFL} admits a unique strong solution in $\ol C_p([0,T], \Pi)$ for all $T>0$. In particular, if $C_0=0$ then the unique solution lies in $C_p([0,T], \Pi)$ for all $T>0$. Moreover,  each Nash equilibrium $\bar\pi^*$ defined in \cref{defn:nash} is an invariant measure of  \cref{eq:MFL}.
\end{theorem}

\begin{remark}
{\rm (i)}\q The dependence on $\mu(\cd|\sfy)$ of the function $\nabla_{x^i}  \frac{\d F^i}{\d \nu}(\nu, \mu,  x^i, \sfy) $ is inevitable for some interesting examples such as \cref{nashcor} in the introduction, where under some mild conditions we may compute
\bee
 \nabla_{x^i}  \frac{\d F^i}{\d \nu} \big(\nu^i, (\nu^j)_{j\neq i} , x^i, \sfy\big):=\int \nabla_{x^i} f^i(x^1 ,\cds,x^i, \cds, x^n, \sfy)  \prod_{j\neq i}\nu^{j}(dx^j|\sfy).
\eee
Note that when there is only one player, there is no such dependence.

\ms
\no{\rm (ii)}\q The Lipschitz condition with respect to $\cW_p(\nu, \nu')$ and $\cW_p(\mu, \mu')$ can be replaced by the one with respect to  $\ol\cW_p(\nu, \nu')$ and $\ol\cW_p(\mu, \mu')$. The latter is weaker.  Under such assumption we cannot prove the particular case that the unique solution  lies in $C_p([0,T], \Pi)$ for all $T>0$ when $C_0=0$. In the following analysis of the one-player problem it is crucial for us that the solution is in  $C_p([0,T], \Pi)$, so we prefer to state the Lipschitz condition in its current form.
\end{remark}

%
%
%
%

\begin{theorem}[Uniqueness of invariant measure: Contraction]\label{thm:contraction}
For $i=1,\cdots, n$ and $\mu\in \Pi^{-i}$, let $F^i(\cdot, \mu): \nu \in\Pi^i\mapsto F^i(\nu, \mu)\in \dbR$  satisfy \cref{assum:base}. 
   Assume that
\begin{itemize}
\item  for each $i=1,\cds,n$, the function $\nabla_{x^i} \frac{\d F^i}{\d \nu}$ is Lipschitz continuous in the following sense:
\begin{multline*}
\left| \left(\nabla_{x^i} \frac{\d F^i}{\d \nu}(\bp^i, \bp^{-i},  x^i, \sfy) \right)_{i=1,\cds,n}
- \left(\nabla_{x^i} \frac{\d F^i}{\d \nu}(\bp'^i, \bp'^{-i}, x'^i, \sfy) \right)_{i=1,\cds,n} \right|\\
\le \g \Big(  \ol\cW_1(\bp,\bp') +\cW_1\big(\pi(\cd|\sfy),\pi'(\cd|\sfy)\big)   \Big) + C|x^i -x'^i|,
\end{multline*}
and \cref{eq:uniform0} holds true;
\item  there is a continuous function $ \kappa:(0,+\infty)\rightarrow\dbR $ s.t. $ \limsup\limits_{r\rightarrow+\infty}\kappa(r)<0$, $\int_{0}^{1}r\kappa(r) d r<+\infty  $ and for any $(\bp,\sfy)\in \Pi\times\dbY$ we have for all $x, x'\in \dbR^N$ ($x\neq x'$)
\bee
\sum_{i=1}^n (x^i -  x'^i)\cdot\Big(-\nabla_{x^i} \frac{\d F^i}{\d \nu}(\bp^i, \bp^{-i}, x^i, \sfy) +\nabla_{x^i} \frac{\d F^i}{\d \nu}(\bp^i, \bp^{-i},  x'^i, \sfy)\Big) 
\leq \kappa\left(|x -x'|\right)\left|x - x'\right|^2.
\eee
\end{itemize}
Let $\bp_0, \bp'_0 \in \cP_q(\bar\dbR^d)\cap \Pi$ for some $q>1$ be two initial distributions of the MFL system \cref{eq:MFL}. Then we have
\be\label{eq:contractionresult}
 \ol\cW_1 (\bp_t, \bp'_t)
\le e^{(2\g-c\sigma^2) t}  \frac{2}{\f(R_1)}  \ol\cW_1 (\bp_0, \bp'_0),
\ee
where the coefficients read
\beaa
	& \f(r) = \exp\left( -\frac12\int_{0}^{r}\frac{u\kappa^+(u)}{\sigma^2} d u \right),\q
	c^{-1} = \int_{0}^{R_2}\Phi(r)\f(r)^{-1} d r,
	\q \Phi(r) = \int_{0}^{r}\f(s) d s,\\
	& R_1 := \inf\{ R\geq 0: \kappa(r)\leq0 \text{ for all }r\geq R \},\\
&R_2 := \inf\{ R\geq R_1: \kappa(r)R(R-R_1)\leq-4\si^2\text{ for all }r\geq R \}.
\eeaa
In particular, if $2\g < c\si^2$ (i.e. the MFL system has small dependence on the marginal laws), there is a unique invariant measure in $\cup_{q>1}\cP_q(\bar\dbR^N)\cap\Pi$.
\end{theorem}

%

\subsection{Special case: one player}

When the problem degenerates to the case of a single player, the MFL dynamics becomes a gradient flow and the function $V=F+\frac{\si^2}{2}H$ is a natural Lyapunov  function for the dynamics.

\begin{theorem}[Gradient flow]\label{thm:GD}
Consider a function $F$ satisfying \cref{assum:base} with $p=2$. Let the assumption of \cref{thm:well} hold true, and further assume that
\begin{itemize}
\item  there is $\e>0$ such that  for all $\bp\in \Pi$ and $\sfy\in \dbY$
\be\label{assum:dissipative}
 x \cd \nabla_x \frac{\d F}{\d \nu}(\bp, x, \sfy) \ge  \e |x|^2, \q\mbox{for $|x|$  big enough} ; 
 \ee

\item for all $\bp\in \Pi$ and $\sfy\in \dbY$, the mapping $x\mapsto  \nabla_x \frac{\d F}{\d \nu}(\bp, x, \sfy)$ belongs to $C^3$;

\item for all $\sfy\in \dbY$, the function $(\bp, x)\mapsto \nabla_x \frac{\d F}{\d \nu}(\bp, x, \sfy), \nabla_x^2 \frac{\d F}{\d \nu}(\bp, x, \sfy)$ are jointly continuous.
\end{itemize}
Then we have for $s'>s>0$ 
\be\label{eq:GD}
 V(\bp_{s'}) -V(\bp_s) =-\int_s^{s'}  \int_{\bar\dbR^N} \left| \nabla_{x} \frac{\d F}{\d \nu}(\bar\pi_t,  x,\sfy) + \frac{\si^2}{2} \nabla_{x}\ln\big(\pi_t(x |\sfy)\big) \right|^2  \bp_t (d\bar x) dt
\ee
\end{theorem}

Using an argument, similar to that in \cite{HRSS19, HKR19},  based on the Lasalle's invariant principle, we can show  the following theorem.
\begin{theorem}\label{thm:invariant}
Consider the following statements: 
\begin{itemize}
\item[{\rm (i)}] $\bp_0\in \cup_{q>2}\cP_q(\bar\dbR^N)$;
\item[{\rm (ii.a)}] $\dbY$ is countable;
\item[{\rm (ii.b)}] $\dbY = \dbR^m$, $\sfm$ is absolutely continuous with respect to the Lebesgue measure  and the function $e^{-\frac{2}{\si^2}\frac{\d F }{\d \nu}(\bp, \cd)}\sfm$ is semiconvex in $\bar x$ for any given $\bp\in \Pi$.
\end{itemize}
Let  the assumptions of \cref{thm:GD} hold true. Further assume {\rm (i)}, {\rm (ii.a)} or {\rm (i)}, {\rm (ii.b)}.
Then all the $\cW_2$-cluster points of the marginal laws $(\bar\pi_t)_{t\ge 0}$ of the MFL system \cref{eq:MFL} belong to  the set
\be\label{eq:limitset}
\cI:=\left\{ \bar\pi \in\Pi:~ \nabla_{x} \frac{\d F}{\d \nu}(\bar\pi,  \cd,\sfy) + \frac{\si^2}{2} \nabla_{x}\ln\big(\pi(\cd|\sfy)\big) =0\q\mbox{$\sfm$-a.s.} \right\}.
\ee
\end{theorem}

\begin{remark}[The limit set and the mean-field equilibria on the environment] 
Consider the case where the probability measure on the environment $\sfm$ is atomless. In particular, for a fixed $\sfy\in \dbY$ the probability $\bar\pi\in \Pi$ does not depend on $\pi(\cdot |\sfy )$. Therefore the equation in \cref{eq:limitset} is a sufficient and necessary condition for 
\be\label{nash_x0}
\pi(\cdot|\sfy) = \argmin_{\nu}  \left( \int  \frac{\d F}{\d \nu}(\bar\pi, x ,\sfy) \nu(dx) + \frac{\si^2}{2} H(\nu|{\rm Leb} )\right),
\q\mbox{for $\sfm$-a.s. $\sfy$},
\ee
where $H(\cd|{\rm Leb})$ is the relative entropy with respect to the Lebesgue measure. 
If we view the variable $\sfy$ as the index of the `players', \cref{nash_x0} indicates that all $\bar\pi\in \cI$ are (mean-field) Nash equilibria of the game where the $\sfy$-player aims at:
\bee
\inf_{\nu}  \left( \int  \frac{\d F}{\d \nu}(\bar\pi, x ,\sfy) \nu(dx) + \frac{\si^2}{2} H(\nu|{\rm Leb})\right).
\eee
\end{remark}

\begin{corollary}\label{thm:convex_1P}
If the function $V$ is convex, the limit set $\cI$ is a singleton and thus the marginal laws $(\bar\pi_t)_{t\ge 0}$ converge in $\cW_2$ to the minimizer of $V$. 
\end{corollary}

\section{Applications}\label{sec:app}

\subsection{Dynamic games and deep neural networks}

As mentioned in \cref{eg:dynamicgame}, both discrete-time and continuous-time dynamic games  can be viewed as games on the random environment. 

Take the continuous-time dynamic game as an example, in particular $\dbY=[0,T]$. Consider the controlled process of the $i$-th player 
\be\label{eq:controlledproc}
d\Theta^i_\sfy = \int \f^i \big( x^i, \bar \pi^{-i}, \Theta_\sfy^i, \sfy\big) \pi^i(dx^i | \sfy) d\sfy.
\ee
and his objective function 
\bee
F^i(\bp^i, \bp^{-i}) := \int_0^T c^i \big( x^i, \bar \pi^{-i}, \Theta_\sfy^i, \sfy\big) \pi^i(dx^i|\sfy) d\sfy + g^i(\Theta^i_T). 
\eee
Define the Hamiltonian function $H^i(x^i, \mu, \th^i, \sfy, p) := c^i(x^i, \mu, \th^i, \sfy) + p \cd \f^i (x^i, \mu, \th^i, \sfy) $.  Assume that 
\begin{itemize}
\item the coefficients $\f^i, c^i $ are uniformly Lipschitz in $(x^i, \th^i)$;

\item $\f^i, c^i, g^i$ are continuously differentiable in $\th^i$;

\item  $\nabla_{\th^i} \f^i, \nabla_{\th^i} c^i $ are uniformly Lipschitz in $(x^i, \th^i)$, and $\nabla_{\th^i} g^i$ is uniformly Lipschitz in $x^i$.
\end{itemize}
\no It follows from a standard variational calculus that 
\bee
\frac{\d F^i}{\d \nu}(\nu, \mu,x^i, \sfy) = H^i(x^i, \mu, \Th^i_\sfy, \sfy, P^i_\sfy),
\eee
where $\Th^i$ follows the dynamics \cref{eq:controlledproc} and $P^i$ is the solution to the linear ODE:
\bee
P^i_\sfy = \nabla_\th g(\Th^i_T) + \int_\sfy^T \nabla_\th H^i(x^i, \mu, \Th^i_\sfy, \sfy, P^i_\sfy)  \pi^i(dx^i | \sfy) d\sfy.
\eee
Therefore, according to \cref{thm:contraction}, the Nash equilibrium of this dynamic game can be approximated by the marginal law of the MFL system.

In case the number of players $n=1$, the marginal laws of the MFL system approximates the minimizer of the optimization.  There is a rising interest in modeling the forward propagation of the deep neural networks using a controlled dynamics and in connecting the deep learning to the optimal control problems, see e.g. \cite{CLT19, CRBD19, EHL19,JSS19,  LM19, LMLLY20}.  For the controlled processes in the particular form \cref{eq:controlledproc},  we refer to Section 4 in \cite{HKR19} for the connection between the optimal control problem and the deep neural networks. In particular, we remark that the backward propagation algorithm is simply a discretization of the corresponding MFL dynamics.

\subsection{Linear-convex zero-sum game and GAN}

Consider the zero-sum game between two players, i.e. $F^1 = -F^2 (=:F)$. For $F$ satisfying the assumption in \cref{thm:contraction} so that the contraction result \eqref{eq:contractionresult} holds true, we may use the following MFL system to approximate the unique Nash equilibrium:
\bea
\begin{cases}\label{eq:zerosumMFL2}
 d X^1_t = - \nabla_{x^1} \frac{\d F}{\d \nu} (\bar\pi^1_t , \bar\pi^2_t, X^1_t, Y)dt + \si dW^1_t,\\
 d X^2_t =  \nabla_{x^2} \frac{\d F}{\d \nu} (\bar\pi^2_t , \bar\pi^1_t, X^2_t, Y)dt + \si dW^2_t.
\end{cases}
\eea

Now we consider a particular subclass of the zero-sum games. Assume that $F:(\nu, \mu)\in \Pi^1\times \Pi^2\mapsto F(\nu, \mu)\in \dbR$ is linear in $\mu$ and convex in $\nu$,  and define $\tilde V(\nu,\mu):= F(\nu,\mu) + \frac{\si^2}{2} \big(H(\nu) -H(\mu)\big)$. In particular, 
\begin{itemize}
\item $\frac{\d F}{\d \mu}$ does not depend on $\mu$;
\item the function $\Phi:\nu\mapsto\max_{\mu\in \Pi^{2}} \tilde V(\nu, \mu) - \frac{\si^2}{2} H(\nu)$ is convex.
\end{itemize}
If the Nash equilibrium exists, denoted by $\bar\pi^*$, by the standard argument we have

\bee
\min_{\nu\in \Pi^{1}}\max_{\mu\in \Pi^{2}} \tilde V(\nu, \mu) = \tilde  V(\bar\pi^{*,1}, \bar\pi^{*,2})=\max_{\mu\in \Pi^{2}} \min_{\nu\in \Pi^{1}} \tilde V (\nu, \mu).
\eee
It follows from \cref{thm:FOC} that $\mu^*[\nu] := \argmax_{\mu\in \Pi^{2}} \tilde V(\nu, \mu)$ has the explicit density 
\bee
\frac{ \mu^*[\nu]  (x^2 , d\sfy)}{\sfm(d\sfy)}= C(\nu, \sfy) e^{-\frac{2}{\si^2} \frac{\d F}{\d \mu}(\nu, x^2, \sfy)},
\eee
where $C(\nu, \sfy)$ is the normalization constant. Further, assume that the function $\Phi(\nu) = \tilde V(\nu, \mu^*[\nu])$ satisfies the assumption of \cref{thm:invariant} and recall that $\Phi$ is convex, it follows from \cref{thm:convex_1P} that we may approximate the minimizer $\bp^{*,1}$ using the dynamics:
\be\label{zerosumMFL}
dX_t = -\nabla_x\frac{\d \Phi}{\d \nu} (\bp_t, X_t, Y) dt + \si dW_t. 
\ee
Compared to the dynamics \eqref{eq:zerosumMFL2}, the dynamics \cref{zerosumMFL} enjoys the natural Lyapunov function $\tilde V$, i.e. $ t \mapsto \tilde V( \bar \pi_t ) $ decreases monotonically.
\ms

As an application, the generative adversarial networks (GAN) can be viewed as a linear-convex game. Given a bounded, continuous, non-constant activation function $\f$, consider the parametrized functions 
\be\label{eq:parafunGAN}
\{z\mapsto \dbE[\f(X,z)]: ~{\rm Law}(X) = \nu\in \cP_2(\dbR^{n^2}) \}
\ee
as the options of the discriminators. The regularized GAN aims at computing the Nash equilibrium of the game:
\beaa
&\mbox{Define}\q \tilde V(\nu, \mu):= - \int   \dbE[\f(X,z)]  (\mu-\hat\mu)(dz) - \frac12 \l \Big(\int |z|^2 \mu(dz) - \dbE[|X|^2] \Big) - \frac{\si^2}{2}\big(H(\mu)-H(\nu)\big)&\\
&\begin{cases}
{\rm Generator:} & \sup_{\mu\in \cP_2(\dbR^{n^1})} \tilde V(\nu, \mu) \\
{\rm Discriminator:} & \inf_{\nu\in \cP_2(\dbR^{n^2})} \tilde V(\nu, \mu) 
\end{cases},&
\eeaa
where $\hat\mu\in \cP_2(\dbR^{n^1})$ is the distribution of interest. Indeed, in order to compute the Nash equilibrium of the game, it is appealing to sample the MFL dynamics \eqref{eq:zerosumMFL2} and approximate its invariant measure. However, in order that the contraction result in  \cref{thm:contraction} holds true, it is crucial that the MFL system has small dependence on the marginal laws, which is not necessarily true in the context of GAN.  Here we present another approach, which exploits the particular structure of the linear-convex game. As discussed before, the optimizer of the generator given $\nu\in \cP_2(\dbR^{n^2})$ is explicit and has the density:
\be
\label{generator_explicit}
\mu^*[\nu](z) = C(\nu) e^{-\frac{2}{\si^2} \big( \dbE[\f(X,z)] +\frac{\l}{2} |z|^2 \big)}.
\ee
Further for the potential function $\Phi(\nu):= \tilde V(\nu, \mu^*[\nu]) - \frac{\si^2}{2} H(\nu)$ we have
\be
\label{discrim_formula}
\frac{\d \Phi}{\d \nu}(\nu, x)= -\int \f(x, z) (\mu^*[\nu]-\hat\mu)(dz) + \frac{\l}{2}|x|^2.
\ee
Then the strategy of the discriminator in the Nash equilibrium can be approximated by the MFL dynamics \cref{zerosumMFL}.  In the perspective of numerical realization, note that the law $\mu^*[\nu]$ can be simulated by the MCMC algorithms such as Metropolis-Hastings. 

\ms

In order to illustrate the advantage of the algorithm using the MFL dynamics \cref{zerosumMFL}, here we present the numerical result for a toy example. We are going to use the GAN to generate the samples of the exponential distribution with intensity $1$. In this test, the optimal response of the generator, $\mu^*[\nu]$, is computed via Metropolis algorithm with Gaussian proposal distribution with zero mean and variance optimised according to \cite{metropolis_efficient}. The discriminator chooses parametrized functions among \cref{eq:parafunGAN}, where 
\bee
\f( X, z ) = C ( A z + b )^+,\q  \text{ with } X = ( C, A, b ).
\eee
When we numerically run the MFL dynamics  \cref{zerosumMFL} to train the discriminator, we use a $3000$-particle system, that is, the network is composed of $3000$ neurones, and set its initial distribution to be standard Gaussian.  The other parameters are chosen as follows: $ \si = 0.4 $, $ dt = 0.01 $, $ \l = 0.2 $. \Cref{fig:exponential} shows the training result after $ 60 $ iterations. 
\begin{figure}[h]
\centering
\begin{subfigure}{.5\textwidth}
  \centering
  \includegraphics[width=.9\linewidth]{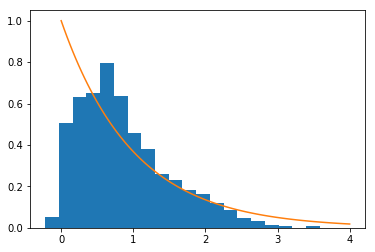}
  \caption{Histogram of GAN's sampling and the target density}
  \label{fig:sub1}
\end{subfigure}%
\begin{subfigure}{.5\textwidth}
  \centering
  \includegraphics[width=.9\linewidth]{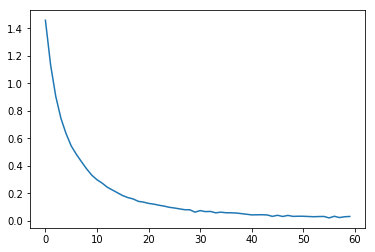}
  \caption{Training error}
  \label{fig:sub2}
\end{subfigure}
\caption{Learning via MCMC-GAN: histogram of learned distribution and training error.}
\label{fig:exponential}
\end{figure}
 In particular, we see that the training error decreases monotonically as suggested by our theoretical results.

\section{Proofs}\label{sec:proof}

\subsection{Optimization with marginal constraint}

\begin{proof}[Proof of \cref{thm:FOC}] \q  \underline{Necessary condition} ~ {\it Step 1.}\q Let $\bar\pi^*\in \Pi$ be a minimizer of $V$. Since $H(\bp^*)<\infty$, the probability measure $\bp^*$ is absolutely continuous wrt ${\rm Leb}\times \sfm$. Take any probability measure $\bp\in \Pi $ such that $H(\bp)<\infty$, in particular $\bp$ is also absolutely continuous wrt ${\rm Leb}\times \sfm$. Denote the convex combination by $\bp^\e: = \e \bp + (1-\e) \bp^* \in \Pi$. Define the function $h(z): = z\ln z$ for $z\in\dbR^+$ and $h(0)=0$. Then
\beaa
0\le \frac{V(\bp^\eps)-V(\bp^*)}{\eps} &=& \frac{F(\bp^\eps)-F(\bp^*)}{\eps} + \eta\int\frac{h(\pi^\eps( x|\sfy))  - h(\pi^*( x|\sfy))}{\e} dx \sfm(d\sfy)  \\
&= &\frac{1}{\eps}\int_{0}^{\eps}\int\frac{\delta F}{\delta \bp}(\bp^\l,\bar x)(\bp-\bp^*)(d\bar x) d\l 
	+ \eta\int\frac{h(\pi^\eps( x|\sfy))  - h(\pi^*( x|\sfy))}{\e} dx \sfm(d\sfy) .
\eeaa
Since $\sup_{\l\in [0,\e]}|\frac{\delta F}{\delta \bp}(\bp^\l,\bar x)| \le C(1+|\bar x|^p)$ and $\bp,\bp^*\in \Pi$, by the dominated convergence theorem
\bee
\lim_{\e\rightarrow 0} \frac{1}{\eps}\int_{0}^{\eps}\int\frac{\delta F}{\delta \bp}(\bp^\l,\bar x)(\bp-\bp^*)(d \bar x) d\l 
= \int\frac{\delta F}{\delta \bp}(\bp^*,\bar x)(\bp-\bp^*)(d \bar x).
\eee
Since the function $h$ is convex, we have $\frac{h(\pi^\eps( x|\sfy))  - h(\pi^*( x|\sfy))}{\e} \le h(\pi( x|\sfy))  - h(\pi^*( x|\sfy))$. Note that $\int \big(h(\pi( x|\sfy))  - h(\pi^*( x|\sfy))\big) d x\sfm(\sfy) = H(\bp)-H(\bp^*)<\infty$. By Fatou lemma, we obtain
\beaa
\limsup_{\e\rightarrow 0}\int\frac{h(\pi^\eps( x|\sfy))  - h(\pi^*( x|\sfy))}{\e} d  x \sfm(d\sfy)
&\le& \int \lim_{\e\rightarrow 0} \frac{h(\pi^\eps( x|\sfy))  - h(\pi^*( x|\sfy))}{\e}  d  x \sfm(d\sfy) \\
&=&\int \ln\pi^*(x|\sfy) (\bp-\bp^*)(d\bar x).
\eeaa
Therefore we have
\be\label{FOCineq}
0\le \limsup_{\e\rightarrow 0}\frac{V(\bp^\eps)-V(\bp^*)}{\eps} 
\le \int \left( \frac{\delta F}{\delta \bp}(\bp^*,\bar x) + \eta\ln\pi^*(x|\sfy)  \right)(\bp-\bp^*)(d\bar x).
\ee

\vspace{3mm}

\noindent {\it Step 2.}\q  We are going to show that for $\sfm$-a.s. $\sfy$ 
\bea\label{eq:nece_constant}
\Xi_{\sfy}(x):= \frac{\delta F}{\delta \bp}(\bp^*,\bar x) + \eta \ln\pi^*( x|\sfy) \q\mbox{is equal to a constant $\pi^*(\cd |\sfy)$-a.s.} 
\eea
Define the mean value $\ol c(\sfy) : = \int_{\dbR^d} \Xi_{\sfy}(x) \pi^*(dx|\sfy)$ and let $\e, \e'>0$. Consider the probability measure $\bp\in \Pi$ absolutely continuous  wrt $\bp^*$ such that
\beaa
\frac{d\pi (\cd|\sfy)}{d\pi^*(\cd|\sfy)} = 
\begin{cases}
1, & \mbox{for $\sfy$ such that} ~\pi^*\big(\Xi_{\sfy}\le \ol c(\sfy) -\e\big| \sfy\big) <\e' \\
\frac{1_{\Xi_{\sfy} \le \ol c(\sfy) -\e}}{\pi^*\big(\Xi_{\sfy} \le \ol c(\sfy) -\e\big| \sfy\big) }, & \mbox{otherwise}
\end{cases}.
\eeaa
Since $\frac{d\bp}{d\bp^*}$ is bounded, we have that $\bp\in \Pi$ and $H(\bp)<\infty$. In particular \cref{FOCineq} holds true for this $\bp$.
Also note that $\Xi_{\sfy} \le \ol c(\sfy)-\e$, $\pi(\cd| \sfy)$-a.s. for $\sfy$ such that $\pi^*\big(\Xi_{\sfy} \le \ol c(\sfy) -\e\big| \sfy\big) \ge \e'$. So we have
\beaa
0&\le &\int_\dbY \int_{\dbR^d} \Xi_{\sfy}(x) \big(\pi(dx|\sfy) -\pi^*(dx |\sfy)\big)\sfm(d\sfy) \\
&=&\int_{\pi^*\big(\Xi_{\sfy}\le \ol c(\sfy) -\e\big| \sfy\big) \ge \e'} \int_{\dbR^d} \Xi_{\sfy}(x) \big(\pi(dx|\sfy) -\pi^*(dx |\sfy)\big)\sfm(d\sfy) \\
& = & \int_{\pi^*\big(\Xi_{\sfy}\le \ol c(\sfy) -\e\big| \sfy\big) \ge \e'}\left( \int_{\dbR^d} \Xi_{\sfy}(x) \pi(dx|\sfy) -\ol c(\sfy)\right)\sfm(d\sfy) \\
&\le& -\e ~\sfm\left\{\sfy:\pi^*\big(\Xi_{\sfy}\le \ol c(\sfy) -\e\big| \sfy\big) \ge \e'\right\}.
\eeaa
Therefore we conclude that
$\pi^*\big(\Xi_{\sfy}\le \ol c(\sfy) -\e\big| \sfy\big) <\e'$ for $\sfm$-a.s. $\sfy$.
Since this is true for arbitrary $\e', \e>0$, we obtain \eqref{eq:nece_constant}.

\vspace{3mm}
\noindent{\it Step 3}.\q We are going to show that $\bp^*$ is equivalent to  ${\rm Leb}\times \sfm$, so that $\Xi_\sfy$ does not depend on $x$, ${\rm Leb}\times \sfm$-a.s. and the first order equation \cref{eq:necessary}  holds true. 
Suppose the opposite, i.e. there is a set $\cK\in \bar\dbR^d$ such that $\bp^*(\cK)=0$ and ${\rm Leb}\times\sfm(\cK)>0$. In particular,  $\ln\pi^*(x|\sfy) = -\infty$ on $\cK$. 
Denote $\cK_\sfy:=\{x\in \dbR^d: (x,\sfy)\in \cK\}$. We may assume that there exist $K>\e>0$ such that $ {\rm Leb}(\cK_\sfy) \in [\e, K]$ for all $\sfy\in \dbY$.
Define a probability measure $\bp\in \Pi$ such that for all Borel-measurable $A\subset \dbR^d$
\bee
\pi (A|\sfy):= \frac12 \pi^*(A|\sfy) + \frac{1}{2 {\rm Leb}(\cK_\sfy)}\int_{A\cap \cK_\sfy} dx.
\eee
It is easy to verify that  $H(\bp)<\infty$, so \cref{FOCineq} holds true and it implies
\beaa
0&\le& \frac12 \int_\cK \left( \frac{\delta F}{\delta \bp}(\bp^*,\bar x) + \eta\ln\pi^*(x|\sfy)  \right)\bp(d\bar x) - \frac12\int \left( \frac{\delta F}{\delta \bp}(\bp^*,\bar x) + \eta\ln\pi^*(x|\sfy)  \right) \bp^*(d\bar x)\\
&\le & -\infty + \int C(1+|\bar x|^p)  \bp^*(d\bar x) -\frac{\eta}{2} H(\bp^*) =-\infty.
\eeaa
It is a contradiction, so $\bp^*$ is equivalent to  ${\rm Leb}\times \sfm$.

\vspace{5mm}

\noindent\underline{Sufficient condition}~ Assume that $F$ is convex. Let $ \bp^*\in \Pi $  satisfy the first order equation \cref{eq:necessary}, in particular, $\bp^*$ is equivalent to ${\rm Leb}\times \sfm$. 
Take any $ \bp\in \Pi  $ absolutely continuous wrt ${\rm Leb}\times \sfm$ (otherwise $ V(\bp)=+\infty $), and thus absolutely continuous wrt the measure $\bp^*$. 
Define $\bp^\eps := (1-\eps)\bp^* + \eps \bp$ for $\eps>0$. 
By the convexity of $F$ we obtain
\beaa
F(\bp) - F(\bp^*)
&\ge& \lim_{\e\rightarrow 0} \frac{1}{\e} \big(F(\bp^\e) - F(\bp^*)\big) \\
&=&\lim_{\e\rightarrow 0} \frac{1}{\e}\int_0^\e \int_{\bar\dbR^d} \frac{\d F}{\d \bp} (\bp^\l, \bar x)  (\bp -\bp^*) (d\bar x)d\l
=\int_{\bar\dbR^d} \frac{\delta F}{\delta \bp}(\bp^*,\bar x ) (\bp -\bp^*) (d\bar x).
\eeaa
The last equality is due to the dominated convergence theorem. On the other hand, by convexity of the function $ h $,
\bee
H(\bp) - H(\bp^*) \ge \int_{\bar\dbR^d} \ln\pi^*(x|\sfy) (\bp -\bp^*)(d\bar x).
\eee
Hence 
\bee
V(\bp)-V(\bp^*)
\geq \int_{\bar\dbR^d} \bigg(\frac{\delta F}{\delta \bp}(\bp^*,\cdot) +\eta \ln \pi^*(x|\sfy)  \bigg)(\bp -\bp^*) (d\bar x) = 0,
\eee
so $\bp^*$ is a minimizer.
\end{proof}

\subsection{Equilibria of game}

\begin{proof}[Proof of \cref{cor:uniq_monot}] \q 
{\rm (i)}\q Let $n=1$. Take three probability measures $\bp,\bp', \bp''\in \Pi$ such that $\bp =\frac12(\bp'+\bp'')$. Denote $\bp'^{\l}:= \l \bp+ (1-\l)\bp'$ and $ \bp''^{\l}:=\l \bp+ (1-\l) \bp''$. By the definition of the linear derivative of $F$ we obtain
\begin{equation*}
F(\bp') - 2 F(\bp)+F(\bp'') = \int_0^1 \int_{\bar\dbR^N} \frac{\d F}{\d \nu} (\bp'^\l, \bar x ) (\bp'-\bp)(d\bar x) d\l
					+ \int_0^1 \int_{\bar\dbR^N} \frac{\d F}{\d \nu} (\bp''^\l, \bar x ) (\bp''-\bp)(d\bar x) d\l.
\end{equation*}
Note that $\bp' - \bp = \bp-\bp'' =\frac{1}{2-2\l} \big( \bp'^\l-\bp''^\l\big)$. Therefore we have
\bee
F(\bp') - 2 F(\bp)+F(\bp'') =  \int_0^1 \frac{1}{2-2\l} \int_{\bar\dbR^N} \left( \frac{\d F}{\d \nu} (\bp'^\l, \bar x )-\frac{\d F}{\d \nu} (\bp''^\l, \bar x ) \right) (\bp'^\l-\bp''^\l) (d\bar x) d\l \geq 0.
\eee
Finally note that $\l\mapsto (F(\bp'^\l), F(\bp''^\l))$ is continuous. So $F$ satisfying the monotonicity condition must be convex on $\Pi$.

On the other hand, suppose $F$ is convex on $\Pi$.  Following a similar computation, we  obtain
\begin{equation*}
0\le \frac{1}{\e} \Big( F(\bp')- F(\bp'^\e) - F(\bp''^\e) + F(\bp'') \Big)\\
= \frac{1}{\e}\int_0^\e \frac{1}{2} \int_{\bar\dbR^N}  \left( \frac{\d F}{\d \nu} (\bp'^\l, \bar x )-\frac{\d F}{\d \nu} (\bp''^\l, \bar x ) \right) (\bp'-\bp'') (d\bar x) d\l.
\end{equation*}
Let $\bp', \bp''\in \Pi$. It follows from the dominated convergence theorem that
\beaa
0 &\le& \lim_{\e\rightarrow 0}  \frac{1}{\e}\int_0^\e \frac{1}{2} \int_{\bar\dbR^N}  \left( \frac{\d F}{\d \nu} (\bp'^\l, \bar x )-\frac{\d F}{\d \nu} (\bp''^\l, \bar x ) \right) (\bp'-\bp'') (d\bar x) d\l\\
&=& \frac12 \int_{\bar\dbR^N}  \left( \frac{\d F}{\d \nu} (\bp', \bar x )-\frac{\d F}{\d \nu} (\bp'', \bar x ) \right) (\bp'-\bp'') (d\bar x).
\eeaa
So $F$ satisfies the monotonicity condition on $\Pi$.

\vspace{3mm}

\noindent {\rm (ii)}\q Let $ \bar\pi^* \in \Pi$ be a Nash equilibrium of the game.
Then, by \cref{cor:NC_Nash} we have that for every $ i $ there exists a function $ f^i : \dbY \rightarrow \dbR$ 
\bee
\frac{ \delta F^i }{ \delta \nu } ( ( \bar \pi^* )^{i}, ( \bar \pi^* )^{-i}, x^i, \sfy ) + \frac{ \sigma^2 }{ 2 } \ln( ( \pi^* )^i( x^i | \sfy )  ) = f^i( \sfy ),
\q\mbox{for $\sfm$-a.s. $\sfy$.}
\eee
Let $ \bar \pi^*, \bar \pi'^* \in \Pi$ be Nash equilibriums. Then monotonicity condition \cref{eq:monotonicity} implies
\bee
\sum_{ i = 1 }^n \int\limits_{ \bar \dbR^N } \left( - \frac{ \si^2 }{ 2 } \ln( (\pi^*)^i( x^i | \sfy ) ) + \frac{ \si^2 }{ 2 } \ln( (\pi'^*)^{i}( x^i | \sfy ) )  \right) ( \bar \pi^* - \bar \pi'^* )( d \bar x ) \geq 0,
\eee
because the marginal distributions of $ \bar \pi^* $ and $ \bar \pi'^* $ on $ \dbY $ coincide. 
The latter inequality can be rewritten
\beaa
0 &\le & - \sum_{ i = 1 }^n \int\limits_{ \bar \dbR^{ n^i } } \ln \left( \frac{ ( \pi^* )^i( x^i | \sfy )  }{ ( \pi'^* )^i( x^i | \sfy ) } \right) ( ( \bar \pi^*)^i - ( \bar \pi'^* )^i )( d \bar x^i )\\
& = & -\sum_{ i = 1 }^n \left( H( ( \bar \pi^* )^i | ( \bar\pi'^* )^i ) + H( ( \bar \pi'^* )^i | ( \bar \pi^* )^i ) \right).
\eeaa
This is only possible if $ ( \bar \pi^* )^i = ( \bar \pi'^* )^i $ for all $ i = 1, \ldots, n $. 
\end{proof}

\vspace{3mm}

\begin{proof}[Proof of \cref{thm:exist_nash}] \q  For $\bp\in \Pi$ denote $\cR^i(\bp):=\argmin_{\nu\in \Pi^i} V^i(\nu, \bp^{-i})$ for $i=1,\cds,n$, and define
\bee
\cR(\bp):=\{\bp'\in \Pi : ~\bp'^i \in \cR^i(\bp)~\mbox{for $i=1, \cds, n$}\}.
\eee 

\noindent{\it Step 1.}\q First we prove that $\cR(\bp)$ is $\cW_p$-compact. For any $\nu^i \in \cR^i(\bp)$, it follows from the first order equation \cref{eq:necessary} that
\bee
\nu^i( \bar x^i) = C( \nu^i, \bp^{-i}) e^{-\frac{2}{\si^2} \frac{\d F^i}{\d \nu}(\nu^i, \bp^{-i}, \bar x^i)},
\eee
where $C( \nu^i, \bp^{-i})>0$ is the normalization constant so that $\nu^i$ is a probability measure. Take a $p'>p$. The condition \eqref{basic_dissipative} implies that $C( \nu^i, \bp^{-i})$ is uniformly bounded as well as
\bee
\int_{\bar\dbR^{n^i}} |\bar x^i |^{p'} \nu^i(d\bar x^i)
\le C( \nu^i, \bp^{-i}) \int_{\bar\dbR^{n^i}} |\bar x^i |^{p'} e^{- C'|\bar x^i|^q +C} dx^i \sfm(d\sfy).
\eee
So we have
\bee
\ol C^i:= \sup_{\nu^i \in \cR^i(\bp) }\int_{\bar\dbR^{n^i}} |\bar x^i |^{p'} \nu^i(d\bar x^i) <\infty.
\eee
Therefore, $\cR(\bp) \subset  \cE :=\{\nu\in \Pi: \nu^i\in \cE^i \mbox{ for $i=1,\cds,n$}\}$, with $\cE^i:= \{\nu^i\in \Pi^i: \int_{\bar\dbR^{n^i}} |\bar x^i |^{p'} \nu^i(d\bar x^i) \le  \ol C^i\}$, and note that $\cE$ is  $\cW_p$-compact.

%

\vspace{3mm}
\noindent{\it Step 2.}\q We are going to show that the graph of $\bp \in \cE \mapsto \cR(\bp) $ is $\cW_p$-closed,  i.e. given $\bp_m, \bp_\infty\in \Pi\cap \cE$,  $\bp'_m\in  \cR(\bp_m)$ such that $\bp_m\rightarrow\bp_\infty$ in $\cW_p$ and $\bp'_m\rightarrow\bp'_\infty\in  \Pi$, we want to show that $\bp'_\infty\in \cR(\bp_\infty)$. 

Denote the concatenation of two probability measures $\nu^i\in \Pi^i, \mu^{-i}\in \Pi^{-i}$ by
\bee
\nu^i\otimes \mu^{-i} (dx, d\sfy) = \nu^i(x^i|\sfy) \mu^{-i}(x^{-i}|\sfy)dx \sfm (d\sfy).
\eee
Note that for $\bp, \bp'\in \cE$ we have $\bp'^{i}\otimes \bp^{-i}\in \cE$. Since $\cE$ is $\cW_p$-compact, there is a subsequence, still denoted by $(\bp_m ,\bp'_m)$, and $ \bp^*\in \Pi\cap\cE$ such that $\bp'^{i}_m\otimes \bp^{-i}_m \rightarrow \bp^*$  in $\cW_p$ and $\bp^{*,i} = \bp'^{i}_\infty$, $\bp^{*,-i} = \bp^{-i}_\infty$. By the lower-semicontinuity of the mapping: $\bp\mapsto V^i(\bp^i, \bp^{-i})$, we have
\be\label{eq:estimatelow}
V^i(\bp'^i_\infty, \bp^{-i}_\infty) = V^i(\bp^{*, i}, \bp^{*,-i}) \leq \liminf_{m\rightarrow\infty} V^i (\bp'^i_m, \bp^{-i}_m).
\ee
Further, fix $\nu^i\in \Pi^i\cap\cE^i$. 
 Again by the compactness of $\cE$, there is a subsequence, still denoted by $(\nu^{i} ,\bp'_m)$, and $ \bp^\nu\in \Pi\cap\cE$ such that $\nu^{i}\otimes \bp^{-i}_m \rightarrow \bp^\nu$  in $\cW_p$ and $\bp^{\nu,i} = \nu^i$, $\bp^{\nu,-i} = \bp^{-i}_\infty$. Therefore
\bee
\liminf_{m\rightarrow\infty} V^i (\bp'^i_m, \bp^{-i}_m)
\le \liminf_{m\rightarrow\infty} V^i (\nu^{i}, \bp^{-i}_m)
= V^i (\bp^{\nu, i}, \bp^{\nu,-i})
= V^i (\nu^i, \bp^{-i}_\infty).
\eee
Together with \cref{eq:estimatelow}, we conclude that $\bp^{'i}_\infty \in \cR^i(\bp_\infty)$ for all $i$, and thus $\bp'_\infty\in \cR(\bp_\infty)$.

\vspace{3mm}
\noindent{\it Step 3.}\q From the condition of the theorem and the result of {\it Step 1\&2}, we conclude that for any $\bp\in \Pi$ the set $\cR(\bp)$ is non-empty, convex, that the set $\cup_{\bp\in \Pi} \cR(\bp) $ is a subset of a $\cW_p$-compact set, and that the graph of the mapping $\bp\in \cE \mapsto \cR(\bp) \subset \cE$ is $\cW_p$-closed. Therefore, it follows from the Kakutani fixed-point theorem that the mapping $\bp \mapsto \cR(\bp)$ has a fixed point, which is a Nash equilibrium.
\end{proof}

\subsection{Invariant measure of MFL system}

\begin{proof}[Proof of \cref{thm:well}] \q
The proof is based on the Banach fixed point theorem.  Given fixed $(\bp_t), (\bp'_t) \in \ol C_p([0,T], \Pi) $, the SDEs \cref{eq:MFL} apparently have unique strong solutions, denoted by $\bar X, \bar X'$ respectively. Denote by 
 $\hat\bp_t: = {\rm Law}(\bar X_t)$ and $\hat\bp'_t: = {\rm Law}(\bar X'_t)$. We are going to show that the mapping $\Psi: \bp\mapsto \hat\bp$ is a contraction as $T$ is small enough. 
 
First, by the standard SDE estimate we know $\dbE\big[\sup_{t\in [0,T]} |X_t|^p \big]<\infty$, and this implies $(\hat\bp_t), (\hat\bp'_t) \in \ol C_p([0,T], \Pi) $. Note the SDEs for $\tilde X, \tilde X'$ share the same Brownian motion.  Denoting $\d X = X -X'$, we have
\beaa
 |\d X^i_s|^p &=& \left|\int_0^s \left( - \nabla_{x^i} \frac{\d F^i}{\d \nu} (\bar\pi^i_t , \bar\pi^{-i}_t, X^i_t, Y)    + \nabla_{x^i} \frac{\d F^i}{\d \nu} (\bar\pi'^i_t , \bar\pi'^{-i}_t, X'^i_t, Y)dt \right)dt \right|^p\\
 &\le & T^{1-\frac1p} \Big(\int_0^s C\Big(\cW_p^p(\bp^i_t, \bp'^i_t)+  \cW_p^p(\bp^{-i}_t, \bp'^{-i}_t) + |\d X^i_{t'}|^p\Big)dt +\int_0^s C_0  \cW_p^p(\bp^{-i}_t(\cd|Y), \bp'^{-i}_t(\cd|Y)) dt \Big).
\eeaa
Taking expectation on both sides, by the Gronwall inequality we obtain that 
\bee
\dbE\big[|\d X^\sfy_s|^p \big] \le C e^{CT}\int_0^T \Big( \cW^p_p(\bp_t, \bp'_t) + \cW^p_p(\bp_t(\cd|\sfy), \bp'_t(\cd|\sfy))\Big)dt.
\eee
 Note that
$\dbE\big[ |\d X^\sfy_s|^p \big] 
\ge\cW_p^p (\hat\bp_s(\cd|\sfy), \hat\bp'_s(\cd|\sfy))$.
Therefore 
\bee
\cW_p^p (\hat\bp_s(\cd|\sfy), \hat\bp'_s(\cd|\sfy)) \le   C e^{CT}\int_0^T \Big( \cW^p_p(\bp_t, \bp'_t) + \cW^p_p(\bp_t(\cd|\sfy), \bp'_t(\cd|\sfy))\Big)dt.
\eee
Integrating both sides with respect to $\sfm$, we obtain
\bee
\ol\cW_p^p (\hat\bp_t, \hat\bp'_t) \le   C e^{CT}\int_0^T 2\ol \cW^p_p(\bp_t, \bp'_t) dt
\le 2CT e^{CT}\sup_{t\in [0,T]} \ol \cW^p_p(\bp_t, \bp'_t) ,
\eee
 and $\Psi$ is a contraction whenever $T$ is small enough. In case $C_0=0$ the result can be deduced similarly.
\end{proof}

\ms

In order to prove \cref{thm:contraction}, the main ingredient is the reflection coupling  in Eberle \cite{eberle11}. For this mean-field system, we shall adopt the reflection-synchronous coupling as in \cite{eberle2019quantitative}. 

We first recall the reflection-synchronous coupling. Fix a parameter $\e>0$. Introduce the Lipschitz functions ${\rm rc}: \dbR^N \times \dbR^N\rightarrow [0,1]$ and ${\rm sc}: \dbR^N \times \dbR^N\rightarrow [0,1]$ satisfying
\beaa
{\rm sc}^2 (x, x') + {\rm rc}^2(x,x') =1, \q
{\rm rc}(x,x') =1 ~~\mbox{for $|x-x'|\ge \e$}, \q 
{\rm rc}(x,x') =0 ~~\mbox{for $|x-x'|\le \e/2$}.
\eeaa
 Let $\bp_0,\bp'_0\in \Pi\cap \cP_q(\bar\dbR^N)$ with some $q>1$ be two initial distributions of the MFL system \cref{eq:MFL}, and $W, \tilde W$ be two independent $N$-dimensional Brownian motions. It follows from \cref{thm:well} the two MFL systems have strong solutions, and denote the marginal laws by $\bp_t, \bp'_t$. Denote the drift of the dynamics \cref{eq:MFL_sys} by   
 \be\label{eq:drift}
 b^\sfy(t,x) := \left(- \nabla_{x^i} \frac{\d F^i}{\d \nu}(\bp^i_t, \bp^{-i}_t, x^i, \sfy)\right)_{i=1,\cds,n},
 \q \tilde b^\sfy(t,x) := \left(- \nabla_{x^i} \frac{\d F^i}{\d \nu}(\bp'^i_t, \bp'^{-i}_t, x^i, \sfy)\right)_{i=1,\cds,n}
 \ee
 Further, for a fixed $\sfy\in \dbY$,  define the coupling $\Si^\sfy=(X^\sfy, X'^\sfy)$ as the strong solution to the  SDE
 \beaa
 && dX^\sfy_t = b^\sfy(t, X^\sfy_t ) dt + \rc (\Si^\sfy_t) \si dW_t + \syc(\Si^\sfy_t)\sigma d\tilde W_t,\\
 && dX'^\sfy_t  = \tilde b^\sfy(t, X'^\sfy_t) dt + \rc (\Si^\sfy_t) \big({\rm Id} - 2 e_t \langle e_t, \cd\rangle \big)\sigma dW_t + \syc(\Si^\sfy_t)\sigma d\tilde W_t,
 \eeaa
where $e_t:= \frac{X^\sfy_t -X'^\sfy_t}{|X^\sfy_t -X'^\sfy_t| }$ for $X^\sfy_t \neq X'^\sfy_t $, otherwise $e_t:=\hat e$ some arbitrary fixed unit vector in $\dbR^N$.

In the following proof, we will use the concave increasing function $f$ constructed in  \cite[Theorem 2.3]{eberle2019quantitative}:
\begin{equation*}
	f(r) := \int_{0}^{r}\f(s)g(s\wedge R_2) d s,\q
	\mbox{where}\q g(r):= 1-\frac{c}{2}\int_0^r \Phi(s)\f(s)^{-1}ds,
\end{equation*}
and the function $ \f  $ and  the constants $R_2, c$ are defined  in the statement of \cref{thm:contraction}.
In particular, on  $(0,R_2)\cup(R_2,+\infty)$ the function $f$ satisfies
\bea
& 2\sigma^2 f''(r)  \le -r \big(\k(r)+2\g \eta(r)\big) f'(r) - c\sigma^2 f(r) \label{eq:fODE}
\eea
and for $ r\in\dbR^+ $
\bea
& r\f(R_1) \le \Phi (r) \le 2 f(r) \le 2\Phi (r) \le 2r. \label{eq:bound_f}
\eea

\ms
\begin{proof}[Proof of \cref{thm:contraction}] \q Note that $\pi_0(\cd|\sfy),\pi'_0(\cd|\sfy) \in \cP_q(\dbR^N)$ with $q>1$ for $\sfm$-a.s. $\sfy\in \dbY$. For such $\sfy$, we may choose the coupling $ ( X^\sfy_0, X'^\sfy_0 ) $ so that
\bea\label{eq:fWlower}
\cW_1 \big(\pi_0(\cd|\sfy),\pi'_0(\cd|\sfy) \big) = \dbE \big[|X^\sfy_0 - X'^\sfy_0|\big] 
\ge \dbE \Big[f\big(|X^\sfy_0-X'^\sfy_0|\big)\Big]  .
\eea
The last inequality is due to \eqref{eq:bound_f}. On the other hand, for all $t\ge 0$ we have
\bea\label{eq:fWupper}
\cW_1 \big(\pi_t(\cd|\sfy),\pi'_t(\cd|\sfy) \big) \le  \dbE \big[|X^\sfy_t - X'^\sfy_t|\big]
\le  \frac{2}{\f(R_1)}\dbE \Big[f\big(|X^\sfy_t - X'^\sfy_t|\big)\Big]. 
\eea
Denote $\d X^\sfy_t : = X^\sfy_t - X'^\sfy_t $. By the definition of the coupling above, we have
\beaa
d \d X^\sfy_t = \Big(b^\sfy(t, X^\sfy_t) - \tilde b^\sfy(t, X'^\sfy_t)\Big) dt + 2 \rc(\Si^\sfy_t) \sigma e_t d\bar W_t,
\eeaa
where $\bar W_t:=\int_0^t  e_s \cd dW_s $ is a one-dimensional Brownian motion. Denote $r_t := |\d X^\sfy_t|$ and note that by the definition of $\rc$ we have $\rc (\Si^\sfy_t) = 0$ whenever $r_t\le \e/2$. Therefore, 
\beaa
d r_t = e_t \cd \Big( b^\sfy(t, X^\sfy_t)- \tilde b^\sfy(t, X'^\sfy_t)\Big) dt + 2 \rc(\Si^\sfy_t) \sigma d \bar W_t. 
\eeaa
Then it follows from the It\^o-Tanaka formula and the concavity of $f$ that
\beaa
 f(r_t) - f(r_0)
&\le & \int_0^t \left( f'(r_s) e_s \cd \big( b^\sfy(s,X^\sfy_s)- \tilde b^\sfy(s, X'^\sfy_s)\big) + 2\rc(\Si^\sfy_s)^2 \sigma^2 f''(r_s) \right) ds
+ M_t,
\eeaa
where $ M_t:= 2 \int_0^t \rc(\Si^\sfy_s) f' (r_s) \sigma d\bar W_s $ is a martingale. Now  note that 
\beaa
  e_s \cd \big( b^\sfy(s,X^\sfy_s)- \tilde b^\sfy(s, X'^\sfy_s)\big)
 \le 1_{\{ r_s\ge \e\}} r_s \k(r_s) + 1_{\{r_s<\e\}} C \e + \g \big(\ol\cW_1(\bp_s, \bp'_s) + \cW_1(\pi_s(\cd|\sfy), \pi'_s(\cd|\sfy)) \big).
\eeaa
Further, since $f'' \le 0$ and $\rc(\Si^\sfy_s) = 1$ whenever $r_s \ge\e$, we have
\beaa
f(r_t) - f(r_0) 
&\le& \int_0^t \Big( 1_{\{ r_s\ge \e\}}\big( f'(r_s) r_s \k(r_s)  + 2 \sigma^2 f''(r_s)\big)  +  1_{\{ r_s < \e\}} C \e \\
        && \q\q\q\q\q\q\q + \g f'(r_s) \big(\ol\cW_1(\bp_s, \bp'_s) + \cW_1(\pi_s(\cd|\sfy), \pi'_s(\cd|\sfy)) \big)\Big) ds  + M_t.
\eeaa
By taking expectation on both sides, we obtain
\beaa
&&\int \dbE \Big[f(r_t) - f(r_0)\Big] \sfm(d\sfy) \\
&\le& \int_0^t \Big( \int \dbE\Big[ 1_{\{ r_s\ge \e\}}\big( f'(r_s) r_s \k(r_s)  + 2 \sigma^2 f''(r_s)\big) +  C \e  \Big] \sfm(d\sfy)  + 2 \g  f'(r_s)\ol \cW_1(\bp_s, \bp'_s) \Big) ds\\
&\le & \int_0^t \Big( \int  \dbE\Big[  -c \si^2 f(r_s)1_{\{ r_s\ge \e\}}  +  C \e   \Big]  \sfm(d\sfy) + 2 \g  \ol \cW_1(\bp_s, \bp'_s)  \Big)ds\\
&\le & \int_0^t \Big( \int  \dbE\Big[  -c \si^2 f(r_s) +  (C + c\si^2) \e   \Big]  \sfm(d\sfy) + 2 \g  \ol \cW_1(\bp_s, \bp'_s)  \Big)ds.
\eeaa
The second  last inequality is due to \eqref{eq:fODE}.
Together with  \eqref{eq:fWlower} and \eqref{eq:fWupper}, we obtain
\beaa
\frac{\f(R_1)}{2}e^{c\sigma^2 t} \ol\cW_1 \big(\bp_t,\bp'_t \big) - \ol\cW_1\big(\bp_0,\bp'_0\big)
&\le & e^{c\si^2t}\int \dbE \big[f(r_t) \big] \sfm(d\sfy)-\int \dbE \big[ f(r_0)\big] \sfm(d\sfy)\\
&\le& \int_0^t e^{c\sigma^2 s} \big( ( C + c \sigma^2 ) \e  +2 \g  \ol \cW_1(\bp_s, \bp'_s) \big)ds .
\eeaa
This holds true for all $\e>0$, so finally we obtain \cref{eq:contractionresult} by Gronwall's inequality.
\end{proof}

\subsection{One player case}

Throughout this subsection, we suppose that the assumptions in \cref{thm:GD} hold true. Recall the drift function $b^\sfy$ defined in \cref{eq:drift} with $n=1$, i.e.
\bee
b^\sfy(t,x):=- \nabla_{x} \frac{\d F}{\d \nu}(\bp_t, x, \sfy).
\eee
Under the assumption of \cref{thm:GD} the function $b^\sfy$ is continuous in $(t,x)$ and $C^3$ in $x$ for all $t\in [0,T]$. Due to a classical regularity result in the theory of linear PDEs (see e.g. \cite[p.14-15]{JKO98}), we obtain the following result.
\begin{lemma}\label{lem:FP}
Under the assumption of \cref{thm:GD},  the marginal laws $(\pi_t(\cd | \sfy))_{t\ge 0}$ of the solution to \cref{eq:MFL_sys}  are weakly continuous solutions  to the Fokker-Planck equations:
\bea\label{eq:FoP}
\partial_t \nu = \nabla_x \cd (-b^\sfy \nu  + \frac{\si^2}{2} \nabla_x \nu) \q\mbox{for}\q \sfy\in \dbY.
\eea
In particular, we have that $(t,x)\mapsto \pi_t(x|\sfy)$ belongs to $C^{1,2}\big((0,\infty)\times \dbR^N)\big)$.
\end{lemma}

The following results can be proved with the same argument as in Lemma 5.5-5.7 in \cite{HKR19}, so the proof is omitted.

\begin{lemma}\label{lem:finit_ent}
 Fix a $\sfy\in \dbY$ and assume $ \dbE\big[ |X^\sfy_0|^2 \big] < \infty $, where we recall the $X^\sfy$ defined in \cref{eq:MFL_sys}. Denote by $ \dbQ^{\sigma}_\sfy$ the scaled Wiener measure\footnote{Let $B$ be the canonical process of the Wiener space and $\dbQ$ be the Wiener measure, then the scaled Wiener measure $ \dbQ^{\sigma}: = \dbQ\circ (\si B)^{-1} $.} with initial distribution $ \pi_0(\cd|\sfy)$ and by $(\cF_t)_{t\ge 0}$ the canonical filtration of the Wiener space. Then
 \begin{enumerate}
 	\item[{\rm (i)}] for any finite horizon $T>0$, the law of the solution to \cref{eq:MFL}, $ \mathbf{\Pi}(\cd|\sfy):={\rm Law}(X^\sfy)$, is equivalent to $ \dbQ^{\sigma}_\sfy$  on $ \cF_T $ and the relative entropy
 	\bee
 		\int \ln \Big( \frac{d\mathbf{\Pi}(\cd|\sfy)}{d\dbQ^{\sigma}_\sfy}\Big|_{\cF_T} \Big) d \nu_t 
		= \dbE\Big[ \int_0^T \big| b^\sfy(t, X^\sfy_t) \big|^2 dt \Big] 
		<+\infty.
 	\eee
 	
 	\item[{\rm (ii)}] the marginal law $ \pi_t(\cd|\sfy)$ admits a density s.t. $\pi_t(\cd|\sfy)>0 $ and $ H\big(\pi_t(\cd|\sfy)\big)<+\infty $ for $t>0$.
	
	\item[{\rm (iii)}]  the function $\ln \pi_t(x|\sfy)$ is continuous differentiable in $x$ for $t>0$, and for any $t_0\in (0,t]$ it satisfies 
	\bee
 			\nabla_x \ln \pi_t(x|\sfy) = -\frac1{t_0}\dbE\left[ \int_{0}^{t_0}\big(1-s \nabla_x b^\sfy(s,X^\sfy_{t-t_0+s})\big)dW^{t-t_0}_s \Big| X^\sfy_t=x \right],
 		\eee
 		where $ W^{t-t_0}_s := W_{t-t_0+s} - W_{t-t_0} $ and $W$ is the Brownian motion in \cref{eq:MFL_sys}.
In particular,  for any $ t_*>0 $ we have 
\beaa
C:=\sup_{s\ge t_*}\int_{\dbR^N } \big| \nabla_x \ln \pi_s(x|\sfy)  \big|^2 \pi_s(x|\sfy)  dx <+\infty,
\eeaa
and $C$ only depends on $t_*$  and the Lipschitz constant of $\nabla_x b^\sfy$ with respect to $ x $. 
	
	\item[{\rm (iv)}] we have the estimates
\beaa
		&\int_{\dbR^N} |\nabla_x \ln \pi_t(x|\sfy) | dx <+\infty, \q \int_{\dbR^N}|x\cdot\nabla_x \ln \pi_t(x|\sfy)| dx<+\infty\q\mbox{for all}\q t>0,\\
		&\mbox{and}\q \int_t^{t'} \int_{\dbR^N}|\D_{xx} \pi_s(x|\sfy)| dx ds <+\infty\q\mbox{for all}\q t'>t>0,
\eeaa
	and together with the integration by parts we obtain for all $t'>t>0$
\bea\label{lem:IPP}
		&\int_{\dbR^N} \D_{xx}  \frac{\d F}{\d \nu}(\bp_t, x, \sfy) \pi_t(x|\sfy) dx 
		= -\int_{\dbR^N} \nabla_{x} \frac{\d F}{\d \nu}(\bp_t, x, \sfy) \cd \nabla_x \pi_t(x|\sfy)dx,\\
		&\int_t^{t'}\int_{\dbR^N}\D_{xx} \big(\ln \pi_s(x|\sfy) \big)\pi_s(x|\sfy) dxds = -\int_t^{t'} \int_{\dbR^N}\left|  \nabla_x \ln \pi_s(x|\sfy) \right|^2 \pi_s(x|\sfy) dxds.\notag
\eea
 \end{enumerate}
 \end{lemma}
 
 The proof of \cref{thm:GD} is based on the previous lemma and  It\^o calculus.
 
 \ms
 
\begin{proof}[Proof of \cref{thm:GD}] \q
It follows from  the It\^o-type formula \cite[Theorem 4.14]{Carmona+Rene} and  \eqref{lem:IPP} that
\bea
\mathrm{d}F(\bp_t) 
&=& \int_{\dbR^N} \left(- \big|\nabla_x \frac{\d F}{\d \nu}(\bp_t,\bar x) \big|^2 + \frac{\sigma^2}{2}{\rm Tr}( \nabla^2_x \frac{\d F}{\d \nu}(\bp_t,\bar x))\right) \bp_t (d\bar x) d t \nonumber \\
&=& \int_{\dbR^N} \left(-|\nabla_x \frac{\d F}{\d \nu}(\bp_t, \bar x)|^2 - \frac{\sigma^2}{2}\nabla_x \frac{\d F}{\d \nu}(\bp_t,\bar x)\cd \frac{\nabla_x \pi_t(x|\sfy)}{\pi_t(x|\sfy)}\right) \bp_t(dx)d t. \label{FreeEnergyEq01}
\eea	
On the other hand, by It\^o's formula and  the Fokker-Planck equation \eqref{eq:FoP}, we have
\bee
 \mathrm{d} \log\pi_t(X_t|\sfy) = \left(\sigma^2\nabla_x \cd \frac{\nabla_x \pi_t(X_t|\sfy)}{\pi_t(X_t|\sfy)} + \frac{\si^2}{2}\left| \frac{\nabla_x \pi_t(X_t|\sfy)}{\pi_t(X_t|\sfy)}\right|^2 - \nabla_x\cdot \nabla_x\frac{\d F}{\d\nu}(\bp_t, X_t, \sfy) \right) dt+ \mathrm{d} M_t .
\eee
where $M$ is a martingale on $[t,T]$ for any $0<t<T$. By taking expectation on both sides and using  \eqref{lem:IPP}, we obtain for $t>0$:
\bea	\label{FreeEnergyEq02}
\mathrm{d}H(m_t) 
&=&\dbE \left[ - \frac{\si^2}{2}\left|\frac{\nabla_x \pi_t(X_t|\sfy)}{\pi_t(X_t|\sfy)} \right|^2 +  \nabla_x\frac{\d F}{\d\nu}(\bp_t, X_t, \sfy)\cd \frac{\nabla_x \pi_t(X_t|\sfy)}{\pi_t(X_t|\sfy)} \right]\mathrm{d}t \notag \\
&=& \int_{\dbR^N} \left( - \frac{\si^2}{2}\left|  \frac{\nabla_x \pi_t(x|\sfy)}{\pi_t(x|\sfy)} \right|^2 - \nabla_x\frac{\d F}{\d\nu}(\bp_t, x, \sfy) \cd \left( \frac{\nabla_x \pi_t(x|\sfy)}{\pi_t(x|\sfy)}\right) \right)\bp(\bar x) dt.
\eea
Summing up equations \eqref{FreeEnergyEq01} and \eqref{FreeEnergyEq02}, we obtain \cref{eq:GD}.
\end{proof}

 \ms

   Let $ (\bp_t)_{t\ge 0} $ be the flow of marginal laws of the solution of \cref{eq:MFL} given an initial law $ \bp_0$. Define a dynamic system $ \cS(t)\left[\bp_0\right] := \bp_t $. 
Due to the result of \cref{thm:GD}, we can view the function $V$ as a Lyapunov function of the dynamic system $(\cS_t)_{t\ge 0}$, and then it is natural to prove \cref{thm:invariant} using LaSalle's invariance principle (see the following \cref{prop:Lasalle}). However, $V^\si$ is not continuous (only lower-semicontinuous), in particular, the mapping $t\mapsto V^\si(\bp_t)$ is not a priori continuous at $+\infty$, which makes the proof non-trivial. Here we follow the strategy first developed in \cite{HRSS19} to overcome the difficulty. 
 Define the $ \omega $-limit set:
\begin{equation*}
	\omega(\bp_0) := \Big\{ \bp\in\Pi:\q \text{there exists }t_n\rightarrow+\infty\text{ such that } \cW_2\left(\cS(t_n)\left[\bp_0\right], \bp\right)\rightarrow 0  \Big\}.
\end{equation*}

\begin{lemma}
Assume that   $ \bp^0\in \Pi\cap\cP_q(\bar\dbR^N)$ with $q\ge 2$. Then for the solution  to the MFL system \cref{eq:MFL} we have
\be\label{eq:uniformmoment}
\sup_{t\ge 0} \int_{\bar\dbR^N} |\bar x|^q  \bp_t(d \bar x) <\infty.
\ee
\end{lemma}
\begin{proof}
By It\^o formula, we obtain
\bee
d |X^\sfy_t|^q = q | X^\sfy_t |^{q-2}\left(- 2 X^\sfy_t \cd \frac{\d F}{\d \nu}(\bp_t, X^\sfy_t , \sfy ) +\si^2q(q-2+N) \right) dt + \si q|X^\sfy_t|^{q-2}X^\sfy_t \cd dW_t.
\eee
By the linear growth assumption on $\frac{\d F}{\d \nu}$ and the dissipative condition \cref{assum:dissipative}, there is a constant $M$ such that
\beaa
d |X^\sfy_t|^q  &\le & q | X^\sfy_t |^{q-2}\left( C - \e |X^\sfy_t|^2 1_{\{|X^\sfy_t|\ge M\}} \right) dt + \si q|X^\sfy_t|^{q-2}X^\sfy_t \cd dW_t\\
			&\le & q | X^\sfy_t |^{q-2}\left( (C+\e M^2) - \e |X^\sfy_t|^2  \right) dt + \si q|X^\sfy_t|^{q-2}X^\sfy_t \cd dW_t.
\eeaa
In particular, the constant $C$ above does not depend on $\sfy$. In case $q =2$, by taking expectation on both sides and using the Gronwall inequality, we obtain
\be\label{eq:tight}
\sup_{t\ge 0} \dbE[|X^\sfy_t|^q] \le C\big(1+  \dbE[|X^\sfy_0|^q]\big),
\ee
with $q=2$.
For $q>2$, a similar inequality follows from the induction. The desired result \cref{eq:uniformmoment} follows from integrating with respect to $\sfm$ on both sides of the inequality. 
\end{proof}

\begin{proposition}[Invariance Principle]\label{prop:Lasalle}
	Assume that   $ \bp^0\in \Pi\cap\cP_q(\bar\dbR^N)$ with $q>2$. Then the set $ \omega(\bp_0) $ is non-empty, compact and invariant, that is
	\begin{enumerate}
		\item[{\rm (a)}] for any $ \bp\in\omega(\bp_0) $, we have $ S(t)\left[\bp\right]\in\omega(\bp_0) $ for all $ t\ge 0 $;
		
		\item[{\rm (b)}] for any $ \bp\in\omega(\bp_0) $ and  $t\ge 0$, there exists $ \bp'\in\omega(\bp_0) $ such that $ S(t)\left[\bp'\right]=\bp $.
	\end{enumerate}
\end{proposition}

\begin{proof}
Note that
\begin{itemize}
\item the mapping $\bp_0 \mapsto  \cS(t)\left[\bp_0\right]$ is $\cW_2$-continuous due to the stability on the initial law;
\item the mapping $t\mapsto  \cS(t)\left[\bp_0\right]$ belongs to $C_2\big(\dbR^+, \Pi\big)$, due to \cref{thm:well};
\item the set $\big\{ \cS(t)[\bp_0],~t\ge 0 \big\}$ belongs to a $\cW_2$-compact set, due to \cref{eq:uniformmoment}.
\end{itemize}
The rest follows the standard argument for Lasalle's invariance principle (see e.g. \cite[Theorem 4.3.3]{Henry81} or \cite[Proposition 6.5]{HRSS19}).
\end{proof}

\ms
\begin{proof}[Proof of \cref{thm:invariant}] \q
{\it Step 1.}\q Recall the set $\cI$ defined in \cref{eq:limitset}. We first prove the existence of a converging subsequence towards a  member in $\cI$. Since $\o (\bp_0)$ is $\cW_2 $-compact and $V$ is $\cW_2$-lower semicontinuous, there is $\bp^*\in \argmin_{\bp\in \o(\bp_0)} V(\bp)$. By the backward invariance (b), given $t>0$ there is $\nu\in \o(\bp_0)$ such that $\cS (t)[\nu] = \bp^*$. By \cref{thm:GD}, we have 
\bee
V\big(\cS(t+s)[\nu]\big) \le V\big(\cS(t)[\nu]\big) = V(\bp^*), \q\mbox{for all $s>0$}.
\eee
Further by the forward invariance (a), we know $\cS(t+s)[\nu] \in \o(\bp_0)$, and by the optimality of $\bp^*$ we obtain $V\big(\cS(t+s)[\nu]\big) = V(\bp^*)$. Again by \cref{thm:GD}, we get
\bee
0= \frac{d V\big(\cS(t)[\nu]\big) }{dt} 
=-  \int_{\bar\dbR^N} \left| \nabla_{x} \frac{\d F}{\d \nu}(\bp^*,  x,\sfy) + \frac{\si^2}{2} \nabla_{x}\ln\big(\pi^*(x |\sfy)\big) \right|^2  \bp^* (d\bar x).
\eee
Since $\bp^*=\cS (t)[\nu] $ is equivalent to ${\rm Leb}\times \sfm$ according to \cref{lem:finit_ent} (ii), we have $\bp^*\in \cI$. By the definition of $\o(\bp_0)$, there is a subsequence of $(\bp_t)$ converging towards $\bp^*$.

\ms

\no{\it Step 2 (a).}\q We first prove the result under the assumption (ii.a). Let $(\bp_{t_n})_n$ be a sequence converging to $\bp^*$ in $\cW_2$. 
Due to the estimate \cref{eq:tight} and the fact that $\dbY$ is countable, there is subsequence, still denoted by $(t_n)_n$ such that for each $\sfy\in \dbY$, $\pi_{t_n}(\cd|\sfy)$ converges to a probability measure $\pi^\sfy$ in $\cW_2$.
Then clearly $\pi^\sfy=\pi^*(\cd|\sfy)$ for $\sfm$-a.s. $\sfy$, and thus $\bp_{t_n}(\cd|\sfy)\rightarrow \bp^*(\cd|\sfy)$ in $\cW_2$ for $\sfm$-a.s. $\sfy$. Note that
\bee
\bp^*(x|\sfy)= C \exp\left( -\frac{2}{\si^2}\frac{\d F}{\d \nu}(\bp^*, x,\sfy)\right),
\eee
in particular, $\bp^*(\cd|\sfy)$ is log-semiconcave. By the HWI inequality (see \cite[Theorem 3]{OV00}) we have
\begin{equation}\label{HWI}
\int  \big(\ln \bp_{t_n}(x|\sfy) -\ln  \bp^*(x|\sfy) \big) \bp_{t_n}(dx|\sfy)
\le 
\cW_2 \big( \bp_{t_n}(x|\sfy)  , \bp^*(x|\sfy) \big) \left(\sqrt{\cI_n^\sfy} + C \cW_2 \big(\bp_{t_n}(x|\sfy)  , \bp^*(x|\sfy) \big) \right),
\end{equation}
where $\cI^\sfy_n$ is the relative Fisher information defined as
\beaa
\cI^\sfy_n &: =& \int \Big| \nabla_x \ln\bp_{t_n}(x|\sfy)  - \nabla_x \ln\bp^*(x|\sfy)   \Big|^2 \bp_{t_n}(dx|\sfy) \\
&=&   \int \Big| \nabla_x \ln\bp_{t_n}(x|\sfy)  +\frac{2}{\si^2}\nabla_x \frac{\d F}{\d \nu}(\bp^*, x,\sfy)  \Big|^2 \bp_{t_n}(dx|\sfy)  \\
&\le & 2 \int \big| \nabla_x \ln \bp_{t_n}(x|\sfy)  \big|^2 \bp_{t_n}(d x|\sfy) + C \Big(1+\int |x|^2 \bp_{t_n}(d x|\sfy) \Big),
\eeaa
where the last inequality is due to the linear growth of $\nabla_x \frac{\d F}{\d \nu}$ in $x$.
It follows from \cref{lem:finit_ent} (iii) that $\sup_{n, \sfy} \cI^\sfy_n<\infty$. Integrate both sides of \cref{HWI} with respect to $\sfm$, and obtain
\be\label{HWIest}
H\big(\bp_{t_n} \big|\bp^*\big) \le C \ol\cW_2(\bp_{t_n}, \bp^*)\Big(1+ \ol\cW_2(\bp_{t_n}, \bp^*) \Big).
\ee
The right hand side converges to $0$ as $n\rightarrow\infty$ by the dominated convergence theorem. Therefore,
\beaa
&& \limsup_{n\rightarrow\infty} H(\bp_{t_n}| {\rm Leb}\times \sfm) - H(\bp^*| {\rm Leb}\times \sfm)\\
& = & \limsup_{n\rightarrow\infty} \int \ln\left(\frac{\pi_{t_n}(x|\sfy)}{ \pi^*(x|\sfy)}\right) \bp_{t_n}(d\bar x)+ \int \ln\big(\pi^*(x|\sfy)\big) (\bp_{t_n} -\bp^*)(d\bar x) \\
 & =& \limsup_{n\rightarrow\infty}H\big(\bp_{t_n} \big|\bp^*\big)  ~\le~ 0,
\eeaa
where the last equality is due to the dominated convergence theorem and the last inequality is due to \cref{HWIest}.
Since $H$ is $\cW_2$-lower-semicontinuous, we have $ \lim\limits_{n\rightarrow\infty} H(\bp_{t_n}| {\rm Leb}\times \sfm) = H(\bp^*| {\rm Leb}\times \sfm) $. Together with the fact that $F$ is $\cW_2$-continuous, we have $\lim\limits_{t\rightarrow\infty} V(\bp_t) = V (\bp^*)$. Further by the $\cW_2$-lower-semicontinuity of $V$, we obtain 
\bee
V(\bp) \le \lim_{t'_n\rightarrow\infty} V(\bp_{t'_n})=V(\bp^*), \q\mbox{for all $\bp\in \o (\bp_0)$}.
\eee
Together with the optimality of $\bp^*$, we have $V(\bp) =V(\bp^*)$ for all $\bp\in \o (\bp_0)$. Finally by the invariant principle and \cref{eq:GD}, we conclude that $\o(\bp_0)\subset\cI$. 

\ms

\no{\it Step 2 (b).}\q Similarly we can prove the result under the assumption (ii.b). Let $(\bp_{t_n})_n$ be a sequence converging to $\bp^*$ in $\cW_2$. Note that
\bee
\bp^*(x, y)= C \exp\left( -\frac{2}{\si^2}\frac{\d F}{\d \nu}(\bp^*, x,\sfy)\right)\sfm(\sfy)
\eee
is log-semiconcave due to the assumption. Due to the HWI inequality, we have
\bee
\int  \big(\ln \bp_{t_n}(x|\sfy) -\ln  \bp^*(x|\sfy) \big) \bp_{t_n}(d\bar x) 
\le 
\cW_2 \big( \bp_{t_n} , \bp^* \big) \left(\sqrt{\cI_n} + C \cW_2 \big(\bp_{t_n} , \bp^* \big) \right),
\eee
where $\cI_n$ is the relative Fisher information defined as
\beaa
\cI_n &: =& \int \Big| \nabla_x \ln\bp_{t_n}(x|\sfy)  - \nabla_x \ln\bp^*(x|\sfy)   \Big|^2 \bp_{t_n}(d\bar x) \\
&=&   \int \Big| \nabla_x \ln\bp_{t_n}(x|\sfy)  +\frac{2}{\si^2}\nabla_x \frac{\d F}{\d \nu}(\bp^*, x,\sfy)  \Big|^2 \bp_{t_n}(d\bar x)  \\
&\le & 2 \int \big| \nabla_x \ln \bp_{t_n}(x|\sfy)  \big|^2 \bp_{t_n}(d \bar x) + C \Big(1+\int |x|^2 \bp_{t_n}(d \bar x) \Big).
\eeaa
Again by \cref{lem:finit_ent} (iii) we have $\sup_{n} \cI_n<\infty$. For the rest, we may follow the same lines of arguments in Step 2 (a) to conclude the proof.
\end{proof}

\bibliographystyle{siamplain}
\bibliography{references}
\end{document}